\documentclass[11pt]{article}
\usepackage[utf8]{inputenc}
\usepackage{amsmath,amssymb,bbm}
\usepackage{algorithm}
\usepackage{algorithmicx}
\usepackage{algpseudocode}
\usepackage[whole]{bxcjkjatype}
\usepackage{graphicx}[dvipdfmx]
\usepackage{mdframed}
\usepackage{natbib}
\usepackage{hyperref}
\hypersetup{
    colorlinks = true,
	citecolor=blue,
	linkcolor=blue
}
\usepackage{geometry}
\usepackage{fourier}
\usepackage{authblk}
\usepackage{enumerate}
\usepackage{amsthm}
\usepackage{subcaption}
\usepackage{booktabs}
\usepackage{colortbl}
\usepackage{xcolor}
\usepackage{tabularx}

\theoremstyle{definition}
\newtheorem{theorem}{Theorem}
\newtheorem*{theorem*}{Theorem}
\newtheorem{definition}[theorem]{Definition}

\newtheorem{lemma}[theorem]{Lemma}

\newtheorem*{note*}{Note}

\newtheorem{corollary}{Corollary}
\newtheorem{assumption}{Assumption}

\newcommand{\argmin}{\mathop{\arg\min}}
\newcommand{\bs}{\boldsymbol}
\newcommand{\diff}{\mathrm{d}}

\newcommand{\dist}{\mathop{\rm dist}}
\newcommand{\innerprod}[2]{\langle #1,#2 \rangle}

\newcommand{\best}[1]{\textcolor{blue}{\textbf{#1}}}
\newcommand{\second}[1]{\underline{#1}}

\title{Outlier-robust neural network training: \\ 
variation regularization meets trimmed loss \\ to prevent functional breakdown}
\author[1,2,3]{Akifumi Okuno\thanks{okuno@ism.ac.jp}}
\author[1,4]{Shotaro Yagishita}
\affil[1]{Institute of Statistical Mathematics}
\affil[2]{The Graduate University for Advanced Studies, SOKENDAI}
\affil[3]{RIKEN}
\affil[4]{Joint Support-Center for Data Science Research}

\date{\empty}

\begin{document}

\maketitle

\begin{abstract}
In this study, we tackle the challenge of outlier-robust predictive modeling using highly expressive neural networks. Our approach integrates two key components: (1) a transformed trimmed loss (TTL), a computationally efficient variant of the classical trimmed loss, and (2) higher-order variation regularization (HOVR), which imposes smoothness constraints on the prediction function.
While traditional robust statistics typically assume low-complexity models such as linear and kernel models, applying TTL alone to modern neural networks may fail to ensure robustness, as their high expressive power allows them to fit both inliers and outliers, even when a robust loss is used. To address this, we revisit the traditional notion of breakdown point and adapt it to the nonlinear function setting, introducing a regularization scheme via HOVR that controls the model's capacity and suppresses overfitting to outliers. 
We theoretically establish that our training procedure retains a high functional breakdown point, thereby ensuring robustness to outlier contamination. 
We develop a stochastic optimization algorithm tailored to this framework and provide a theoretical guarantee of its convergence. 
\end{abstract}

\section{Introduction}

Highly non-linear parametric models, such as deep neural networks, have gained significant attention due to their impressive expressive power~\citep{Goodfellow-et-al-2016,sze2017efficient,miikkulainen2019evolving,samek2021explaining}. These models can approximate arbitrary continuous functions~\citep{cybenko1989approximation,yarotsky2017error}, making them highly adaptive to complex target functions. This expressive power enables them to capture intricate patterns, as demonstrated by large language models~(see, e.g., \citet{brown2020language}; \citet{chang2024survey}) in modeling human language structures. However, this flexibility also introduces the risk of overfitting, where the model fits not only the true underlying patterns but also the noise in the training data, particularly in the common scenario where the data size is limited.

In the presence of outliers, overfitting becomes an even more critical concern, as models may learn spurious patterns that do not reflect the underlying data distribution. A common strategy to mitigate this issue is to employ robust loss functions, such as the trimmed loss~\citep{rousseeuw1984least}, which reduce the influence of outliers by discarding a fixed number of the largest sample-wise losses during training. While such approaches are effective in many traditional settings, we observe that robust losses alone are often insufficient when applied to highly expressive models like neural networks, as demonstrated in Figure~\ref{fig:illustration}(\subref{subfig:NN+Huber}) and Figure~\ref{fig:illustration}(\subref{subfig:NN+Tukey}). This observation serves as a central motivation for the present study.

\begin{figure*}[!t]
\centering
\begin{minipage}{0.235\textwidth}
\centering
\includegraphics[width=\textwidth]{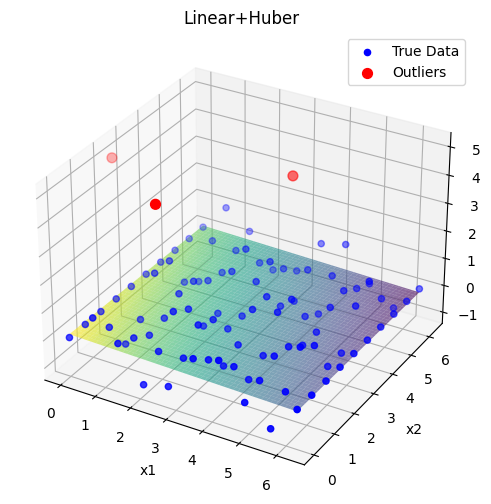}
\subcaption{Linear+Huber}
\label{subfig:Linear+Huber}
\end{minipage}
\begin{minipage}{0.235\textwidth}
\centering
\includegraphics[width=\textwidth]{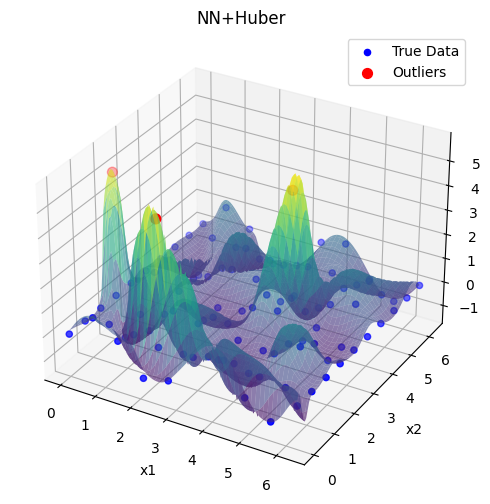}
\subcaption{NN+Huber}
\label{subfig:NN+Huber}
\end{minipage}
\begin{minipage}{0.235\textwidth}
\centering
\includegraphics[width=\textwidth]{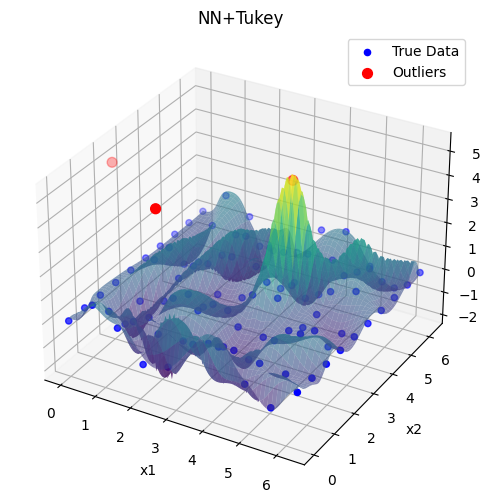}
\subcaption{NN+Tukey}
\label{subfig:NN+Tukey}
\end{minipage}
\begin{minipage}{0.235\textwidth}
\centering
\fbox{
\includegraphics[width=\textwidth]{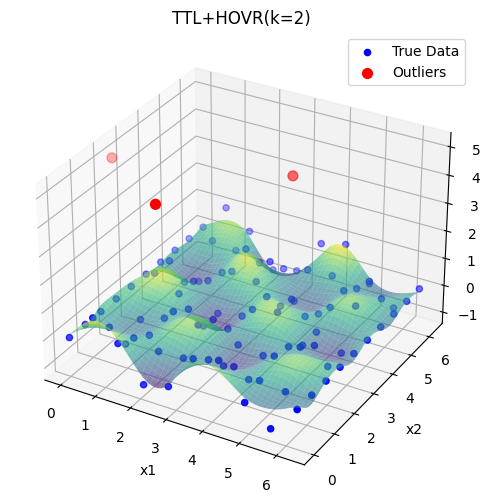}
}
\subcaption{\textbf{Ours}: NN+ARTL}
\label{subfig:Ours}
\end{minipage}
\caption{Illustration of our proposed approach, using the wave-like function $f_*(x)=\sin(2x_1)\cos(2x_2)$. This function is later referred to as the `checkered' function due to the checkered appearance of its contour plot. (\subref{subfig:Linear+Huber}): The linear model is robustly estimated, but it cannot capture complex patterns. (\subref{subfig:NN+Huber}) and (\subref{subfig:NN+Tukey}): Neural networks (NNs) without variation regularization may overfit to outliers, even when using robust loss functions such as Huber's or Tukey's, due to the excessive degrees of freedom in the network. (\subref{subfig:Ours}): Our proposed approach effectively ignores the outliers while preserving the expressive power. See Section~\ref{subsec:synthetic_dataset_experiments} for more details of the synthetic dataset experiments. Note that the simple weight decay, often used to treat singularity in the loss, usually does not address this issue.}
\label{fig:illustration}
\end{figure*}

To address this issue, we begin by revisiting the traditional concept of the \emph{breakdown point} in robust statistics (see, e.g., \citet{hampel1971general} and \citet{donoho1983notion}), which measures the smallest proportion of arbitrarily large outliers that can cause the estimated parameters to diverge.  
Extensions by \citet{alfons2013sparse} suggest that robust loss functions, particularly trimmed losses, combined with parameter regularization can achieve favorable breakdown behavior even in high-dimensional settings.  
Although the above results are restricted to linear models, \citet{alfons2013sparse} can be easily extended to even non-linear settings as also generalized in this study. 
At first glance, this appears to contradict the empirical results shown in Figure~\ref{fig:illustration}; a nonlinear neural network trained with the trimmed loss still suffers severely from the influence of outliers. 
However, the apparent inconsistency stems from a fundamental limitation of traditional robust statistics: its focus is narrowly confined to the stability of parameter estimates.

In models like linear regression or kernel methods, the relationship between parameters and model output is effectively linear, so the breakdown of parameters closely reflects the breakdown of predictions. In contrast, for highly flexible models such as neural networks, this link weakens. In light of these concerns, \citet{stromberg1992breakdown} proposed a variant of the breakdown point based on the model's output values rather than its parameters, specifically targeting nonlinear models. This concept was further generalized by \citet{sakata1995alternative} to encompass a wider class of evaluation metrics, and is herein referred to as the \emph{functional breakdown point}.

Building on the ideas of \citet{stromberg1992breakdown} and \citet{sakata1995alternative}, where they focused on revising the breakdown point definition but did not address estimation procedures, this study proposes a neural network training framework aimed at achieving a high functional breakdown point. Specifically, extending the notion of robustness beyond function values as considered in \citet{stromberg1992breakdown}, we first introduce higher-order variation (HOV), which builds on the concept of total variation~\citep{rudin1992nonlinear, koenker2004penalized, bredies2020higher}. Particularly, we leverage HOV as a regularization to constrain the degrees of freedom of neural networks, thereby controlling their expressiveness. By combining robust trimmed loss with HOV regularization (HOVR), we harness the advantages of both robustness to outliers and adaptability to nonlinear patterns. Furthermore, we demonstrate that this combination achieves a high functional breakdown point defined with HOV, providing a theoretical guarantee of robustness in the presence of outliers.

A final challenge lies in the computational difficulty of optimizing the trimmed loss equipped with HOVR; the trimmed loss requires sorting sample-wise losses, while HOVR involves integrating variations of the neural network, making direct optimization intractable.  
To address this, we propose a stochastic gradient–supergradient descent (SGSD) algorithm to minimize the resulting loss function.  
Building on the theoretical foundations of \citet{robbins1951stochastic}, \citet{ghadimi2013stochastic}, and \citet{bottou2018optimization}, we provide a convergence guarantee for the proposed algorithm. 
As a result, the training procedure becomes both robust to outliers and computationally efficient for neural networks.  
An illustration of our approach is provided in Figure~\ref{fig:illustration}(\subref{subfig:Ours}).

\subsection{Related works}

This section discusses related works on the key concepts of this paper: HOVR and TTL.

\subsubsection*{Works related to HOVR} 
Various heuristic approaches have been proposed to prevent overfitting in neural network training, such as dropout \citep{srivastava2014dropout}, early stopping~(e.g., \citet{yao2007early}), data augmentation~(e.g., \citet{hernandez2018data, chiyuan2021understanding}), and batch normalization \citep{luo2018towards}. These methods are considered forms of regularization and can be combined with parameter regularization, also called weight decay~\citep{kingma2014adam, Goodfellow-et-al-2016}. 

Total variation~(TV) regularization~\citep{rudin1992nonlinear, engl1996regularization, osher2005iterative} has been applied to regression models using splines~\citep{stone1994polynomial, mammen1997locally}, kernels~\citep{zou2005regularization}, and triograms~\citep{koenker2004penalized}, and TV can be regarded as a special case of HOVR with $k=q=1$. 
TV has been extended to second order~\citep{koenker2004penalized, hinterberger2006variational, duan2015second} and general order~($k \in \mathbb{N}$)~\citep{bredies2020higher}, with optimization via discrete approximations. In neural networks, neural splines~\citep{williams2021neural} approximate second-order TV terms through finite approximations. The representer theorem is also studied in Banach space~\citep{parhi2021banach} with TV terms measured in the Radon domain. These theories have been further developed~\citep{parhi2022what, unser2023ridge, bartolucci2023understanding}, but practical optimization algorithms remain underdeveloped.

Neural network training with the proposed HOVR can also be regarded as a special case of the physics-informed training of neural network~\citep[PINN;][]{dissanayake1994neural,berg2018unified,raissi2019pinn,cuomo2022scientific}. 
Our approach can minimize the integral-type regularization more efficiently, without computing the explicit numerical integration, while the conventional PINN implementation finitely approximates the regularization terms over a fixed grid. Also, most PINN leverages non-robust loss functions such as mean squared error.

\subsubsection*{Works related to TTL} 

Robust loss functions have largely been developed following the seminal work of \citet{beaton1974fitting}, \citet{huber1981robust} and \citet{rousseeuw1984least}. Building on this foundation, \citet{yagishita2024fast} introduces an innovative approach that transforms sparse least trimmed squares (SLTS) into a more tractable problem, making it solvable via a deterministic proximal gradient method. This method exhibits strong convergence properties for linear models, outperforming FAST-LTS~\citep{alfons2013sparse} and TSVRG~\citep{aravkin2020trimmed} in speed and offering potentially strong convergence. While these studies are grounded in proximal methods, our study adopts a stochastic approach, while preserving the idea of transforming trimmed loss into a more manageable form.

While convergence for nonlinear models has been addressed, proximal methods can be less efficient, particularly in our case where computing the full gradient of the integral-based HOVR term is intractable. To overcome this, we propose a stochastic approach. The convergence of stochastic proximal gradient-type methods for non-smooth, non-convex problems has been explored by \citet{xu2019stochastic, metel2019simple, xu2019non, metel2021stochastic, yun2021adaptive, deleu2021structured}. However, \citet{xu2019stochastic} and \citet{metel2019simple, metel2021stochastic} assume Lipschitz continuity for the non-convex regularization term, which is not applicable in our context. Additionally, \citet{xu2019non}, \citet{yun2021adaptive}, and \citet{deleu2021structured} require increasing mini-batch size for convergence, which is impractical for real-world use. Note that TSVRG \citep{aravkin2020trimmed} is not suited to our problem setting.

\section{Robust NN Training to Prevent Breakdown}
\label{sec:HOVR}

We describe a multi-layer perceptron, referred to here as a neural network, in Section~\ref{subsec:NN}.  
The robust trimmed loss is introduced in Section~\ref{subsec:trimmed_loss}.  
In Section~\ref{subsec:PBP}, we define the (parametric) breakdown point, a traditional robustness metric, and subsequently propose a novel functional breakdown point in Section~\ref{subsec:FBP}.  
To achieve a high functional breakdown point, we introduce higher-order variation regularization in Section~\ref{subsec:HOVR}.

\subsection{Neural network}
\label{subsec:NN}

Let $J,n \in \mathbb{N}$ be the dimension of the covariate and the sample size, respectively. For each $j=1,2,\ldots,J$, let $D_{j}^{-},D_{j}^{+} \in \mathbb{N}$ define the support of the covariate, with $D_{i}^{-} < D_{i}^+$; we then define the compact covariate space as $\Omega := \prod_{j=1}^{J} [D_{j}^{-},D_{j}^{+}]$. We denote the observations as pairs $(x_i,y_i) \in \Omega \times \mathbb{R}$, for $i=1,2,\ldots,n$. 

We consider training the regression model $f_{\theta}:\Omega \to \mathbb{R}$ parameterized by $\theta \in \Theta \subset \mathbb{R}^r$ ($r \in \mathbb{N}$), such that $f_{\theta}(x_i)$ closely approximates $f_*(x_i)=\mathbb{E}[y_i | x_i]$. For theoretical analysis, we assume that $f_{\theta}$ is sufficiently smooth over $\Theta \times \Omega$, but in practice, it is typically sufficient for $f_{\theta}$ to be smooth almost everywhere. 
For simplicity, this study employs a multi-layer perceptron 
\begin{align*}
    (\text{MLP}): \quad 
    f_{\theta}(x) &= (l_{\theta}^{(Q)} \cdot \tilde{l}_{\theta}^{(Q-1)} \cdot \cdots \cdot \tilde{l}_{\theta}^{(0)})(x),
\end{align*}
where 
$\tilde{l}_{\theta}^{(q)}=\sigma \circ l_{\theta}^{(q)}$, $l_{\theta}^{(q)}(x) = A^{(q)} x + b^{(q)}$ denotes a linear function, and $\sigma:\mathbb{R}\to\mathbb{R}$ represents an activation function (with $\sigma(z)$ applied element-wise if $z$ is a vector). Typical choice for $\sigma$ includes the sigmoid function $\sigma(z)=1/\{1+\exp(-z)\}$ or the hyperbolic tangent function $\sigma(z)=\{\exp(z)-\exp(-z)\}/\{\exp(z)+\exp(-z)\}$. MLP is known as a specific architecture of the neural network having the universal approximation capability~\citep{cybenko1989approximation,yarotsky2017error}, and MLP with large $Q \gg 0$ is also known as a deep neural network~\citep{Goodfellow-et-al-2016}. 

The parameters to be estimated in this MLP are $\theta=(\{A^{(q)},b^{(q)}\}_{q=0}^{Q})$, where $A^{(q)} \in \mathbb{R}^{L_{q+1} \times L_q}$ and $b^{(q)} \in \mathbb{R}^{L_{q+1}}$ are defined with the user-specified number of hidden units $L_{1},L_{2},\ldots,L_{Q} \in \mathbb{N}$. The input and output dimensions are specified as $L_{0}=J$ and $L_{Q+1}=1$. 

While we focus on the MLP in this study, it is worth noting that the follwing discussion can be applied to arbitrary regression models, including traditional splines, kernel regression models, and general architecture of neural networks.

\subsection{Trimmed loss function}
\label{subsec:trimmed_loss}

To train the neural network, this section explains the trimmed loss. 
Let the residual be $r_i(\theta)=y_i-f_{\theta}(x_i)$ and denote the sorted indices by $(1;\theta),(2;\theta),\ldots,(n;\theta)$, where they satisfy $|r_{(1;\theta)}(\theta)| \le |r_{(2;\theta)}(\theta)| \le \cdots \le |r_{(n;\theta)}(\theta)|$. Given a user-specified hyperparameter $h \in \mathbb{N}$ and the residual vector $r(\theta)=(r_1(\theta),r_2(\theta),\ldots,r_n(\theta))$, the $h$-trimmed loss~\citep{rousseeuw1984least} is defined as:
\[
    T_h(r(\theta))
    =
    \frac{1}{n}
    \sum_{i=1}^{h} r_{(i;\theta)}^2(\theta),
\]
and it aims to mitigate the adverse effect of outliers by discarding the $(n-h)$-largest sample-wise losses. 
The $h$-trimmed loss coincides with the standard mean squared error $n^{-1}\sum_{i=1}^{n}\{y_i-f_{\theta}(x_i)\}^2$ by specifying $h=n$. 
Typically, the hyperparameter $h$ is chosen to satisfy $h < n$, allowing the method to discard a fraction of the samples.  
For instance, in our experiments, we set $h = 0.9n$.

\subsection{Parameter breakdown point}
\label{subsec:PBP}

To motivate the functional breakdown point introduced in Section~\ref{subsec:HOVR}, we first review the traditional notion of the (parameter) breakdown point developed in robust statistics, which measures the stability of parameter estimates against contamination by outliers.
Given an inlier dataset $Z = \{(x_i, y_i)\}_{i=1}^{n}$ and a user-specified parameter norm $\|\cdot\|$, we define the traditional parameter-based breakdown point as follows:

\begin{definition}
Let $\sup_{\widetilde{Z}^{(m)}}$ denote the supremum taken over all datasets $\widetilde{Z}^{(m)}$ obtained by replacing $m$ original samples in $Z$ with arbitrary values. Additionally, let $\hat{\theta}(\widetilde{Z}^{(m)})$ denote the estimator trained on $\widetilde{Z}^{(m)}$. Then, the (parameter) breakdown point~\citep{hampel1971general,donoho1983notion} is defined as:
\begin{align*}
    \mathcal{E}(\hat{\theta}, Z)
    := \min\left\{
        \frac{m}{n} ~\middle|~ \sup_{\widetilde{Z}^{(m)}} \|\hat{\theta}(\widetilde{Z}^{(m)})\| = \infty
    \right\}.
\end{align*}
\end{definition}

This definition quantifies the proportion of contamination that the estimator \( \hat{\theta} \) can tolerate before its parameter norm diverges to infinity. A larger breakdown point implies greater robustness, meaning that the estimator remains stable under a higher fraction of outliers.

It is well known that the parameter breakdown point of the linear regression estimator obtained by minimizing the mean squared error,
$\hat{\theta} = \argmin_{\theta} \frac{1}{n} \sum_{i=1}^{n} \{y_i - \langle x_i, \theta \rangle\}^2$, is $1/n$, meaning that even a single outlier is sufficient to break down the parameter estimation.  
However, it has also been shown that the breakdown point can be improved to $(n-h+1)/n$ when minimizing the $h$-trimmed loss instead. $h$ is taken to be greater than roughly half of $n$.  

Thus, the outlier robustness of the trimmed loss can be naturally understood from the perspective of breakdown points.  
Further investigation into the breakdown behavior of parameter estimation with trimmed loss in high-dimensional settings is provided by \citet{alfons2013sparse}.

\subsection{Functional breakdown point}
\label{subsec:FBP}

While traditional robust statistics has mainly focused on linear models or kernel models, these models are classified as basis function models of the form $f_{\theta}^{\dagger}(x) = \sum_{l=1}^{L} \theta_l \phi_l(x)$, which also include spline models.  
In these models, there is an effectively linear relationship between the essential parameters and the output values, so the behavior of the parameters directly and significantly influences the output of the function.  
In contrast, in many nonlinear models especially neural networks, which are the focus of this study, the relationship between the parameters and the model outputs is highly nonlinear. 
Therefore, it is important to monitor the breakdown of the function itself rather than that of the parameters.
From this perspective, \citet{stromberg1992breakdown} proposed a revised definition of breakdown based on the prediction model \( f_{\theta}(x) \) rather than the parameters, which we herein refer to as the \emph{functional breakdown point}. 
While their definition of breakdown is invariant under reparameterization, it may not be appropriate for arbitrary families of nonlinear regression models. \citet{sakata1995alternative} pointed this out and argued that the concept of breakdown should be defined in a manner tailored to each model family.

Inspired by the spirit of \citet{sakata1995alternative}, we consider a variant of function variation and formally define the higher-order variation (HOV) as follows:
\begin{align}
    C_{k,q}(f_{\theta})
    =
    \sum_{\bs i}
    w_{\bs i}
    \int_{\Omega} 
    \left|
    \nabla^{[k]}_{\bs i}
    f_{\theta}(x)
    \right|^q \, \mathrm{d}x,
    \label{eq:HOV}
\end{align}
where \( k \in \{0,1,2,\ldots\} \) and \( q > 0 \) are user-specified parameters, \( w_{\bs i} \geq 0 \) are weights satisfying \( \sum_{\bs i} w_{\bs i} = 1 \), and \( \nabla^{[k]}_{\bs i} = \partial^k / \partial x_{i_1}\partial x_{i_2}\cdots\partial x_{i_k} \) denotes the mixed partial derivative corresponding to multi-index \( \bs i = (i_1, i_2, \ldots, i_k) \). This formulation of HOV measures the \( q \)th powered variation of the \( k \)th order derivatives of the function \( f_{\theta} \).

Given an inlier dataset \( Z = \{(x_i, y_i)\}_{i=1}^{n} \), we define the \emph{functional breakdown point} in terms of HOV as follows:

\begin{definition}
\label{def:FBP}
Let \( k \) and \( q \) be indices as defined in Section~\ref{sec:HOVR}. Let \( \sup_{\widetilde{Z}^{(m)}} \) denote the supremum taken over all datasets \( \widetilde{Z}^{(m)} \) obtained by replacing \( m \) samples in \( Z \) with arbitrary values. Additionally, let \( \hat{\theta}(\widetilde{Z}^{(m)}) \) denote the estimator trained on \( \widetilde{Z}^{(m)} \). Then, the functional breakdown point is defined as:
\begin{align*}
    \mathcal{E}_{k,q}^*(f_{\hat{\theta}}, Z) 
    := \min\left\{
        \frac{m}{n} ~\middle|~ \sup_{\widetilde{Z}^{(m)}} C_{k,q}(f_{\hat{\theta}(\widetilde{Z}^{(m)})}) = \infty
    \right\}.
\end{align*}
\end{definition}

Note that our functional breakdown point is also invariant under reparameterization.
Hereinafter, this study proposes a neural network training framework designed to achieve a high functional breakdown point. Theorem~\ref{theo:breakdown_point} establishes that the functional breakdown point of the proposed framework allows it to tolerate $n-h$ outliers in the dataset.

\subsection{HOVR}
\label{subsec:HOVR}


To attain a high functional breakdown point (see Definition~\ref{def:FBP}), this study proposes incorporating the HOVR to the trimmed loss: we propose training the neural network by minimizing: 
\begin{align}
    \frac{1}{2}T_h(r(\theta))
    +
    \lambda C_{k,q}(f_{\theta}).
    \label{eq:HOV-regularized-TTL}
\end{align}

As the simplest and most illustrative example, for basis function models of the form \( f_{\theta}^{\dagger}(x) = \sum_{l=1}^{L} \theta_l \phi_l(x) \), the parameter regularization \( \|\theta\|_2^2 \) corresponds to the HOVR \( C_{k,2}(f_{\theta}) \), as discussed in Supplement~\ref{app:regularization_kernel}.  
Therefore, in traditional models such as linear regression or kernel regression, parameter regularization and HOVR are essentially compatible.
However, when considering highly expressive models like neural networks, HOVR and parameter regularization become fundamentally different concepts.
While, to the best of the authors' knowledge, there is no existing practical regularization exactly identical to the HOVR, it encompasses many known cases: for instance, HOVR with \( (k,q) = (1,1) \) corresponds to total variation regularization~\citep{rudin1992nonlinear, engl1996regularization, osher2005iterative}, and extensions to second-order~\citep{koenker2004penalized, hinterberger2006variational} and general-order~\citep{bredies2020higher} variations have also been proposed.

Herein, we assume that the regularized trimmed loss function~\eqref{eq:HOV-regularized-TTL} admits a minimizer \( \hat{\theta} \) for any dataset. 
By virtue of the HOVR, we obtain the following theorem, extending the mathematical proof in \citet{alfons2013sparse} to our nonlinear and HOV-based functional breakdown point:

\begin{theorem}
\label{theo:breakdown_point}
Let $f_{\hat{\theta}}$ be the prediction model trained by minimzing the regularized trimmed loss \eqref{eq:HOV-regularized-TTL}. 
The following holds:
\[
\mathcal{E}_{k,q}^*(f_{\hat{\theta}}, Z) \ge \dfrac{n-h+1}{n}.
\]
\end{theorem}

See Supplement~\ref{supp:proof_of_theo:breakdnwon_point} for the proof. Theorem~\ref{theo:breakdown_point} establishes that the NN training framework, which minimizes \eqref{eq:HOV-regularized-TTL}, can tolerate up to $n-h$ outliers while maintaining finiteness of the HOV.

However, a computational challenge remains in optimizing the regularized trimmed loss function~\eqref{eq:HOV-regularized-TTL}. The trimmed loss requires sorting sample-wise losses, and HOVR involves integrating variations of the nonlinear prediction model $f_{\theta}$. Therefore, the next section discusses the proposed stochastic optimization algorithm to address these difficulties.

\section{Stochastic Optimization Algorithm}

To efficiently minimize the HOV-regularized trimmed loss function~\eqref{eq:HOV-regularized-TTL}, this section introduces a stochastic optimization algorithm.  
Specifically, Section~\ref{subsec:SGSD} defines the augmented and regularized trimmed loss (ARTL) and proposes the stochastic gradient–supergradient descent (SGSD) algorithm. Section~\ref{subsec:stochastic_gradient} presents a practical construction of the stochastic gradient used within SGSD. Finally, Section~\ref{subsec:convergence_analysis} provides a theoretical guarantee for the convergence of the proposed algorithm.

\subsection{ARTL and SGSD}
\label{subsec:SGSD}

Since the trimmed loss involves a sorting operation in its computation, it is difficult to handle during optimization. To address this difficulty, \citet{yagishita2024fast} demonstrates that for arbitrary function $f_{\theta}$, including neural networks, minimizing the trimmed loss $T_h(r(\theta))$ is equivalent to minimizing the TTL:
\begin{align}
    \min_{\xi \in \mathbb{R}^{n}}
    \left\{ 
        \frac{1}{n}\|r(\theta)-\xi\|_2^2 + T_h(\xi)   
    \right\}
    =
    \frac{1}{2}T_h(r(\theta)).
    \label{eq:TTL_expression}
\end{align}

While it introduces an additional parameter $\xi$, it is more computationally tractable because the neural network parameter \( \theta \) to be estimated appears outside the sorting operation required by the trimmed loss. 
Considering the equiavlent expression of TTL \eqref{eq:TTL_expression}, minimizing the regularized trimmed loss \eqref{eq:HOV-regularized-TTL} is equivalent to minimizing the augmented and regularized trimmed loss (ARTL) 
\begin{align*}
        F_{h,\lambda}(\theta,\xi)
        =
        U_{\lambda}(\theta,\xi) - V_h(\xi)        
\end{align*}
with respect to the augmented parameter $(\theta,\xi)$, where
\begin{align*}
    U_{\lambda}(\theta,\xi)
    &=
    \frac{1}{n}\{\|r(\theta)-\xi\|_2^2 + \|\xi\|_2^2\}
    +
    \lambda C_{k,q}(f_{\theta}), \\
    V_h(\xi)
    &=
    \frac{1}{n}\|\xi\|_2^2 - T_h(\xi),
\end{align*}
are nonconvex smooth, and convex non-smooth functions, respectively. 
As the former $U_{\lambda}(\theta,\xi)$ is smooth, we can define its stochastic gradient $u_{\lambda}^{(t)}(\theta,\xi)$ that unbiasedly estimates the gradient $\partial U_{\lambda}(\theta,\xi)/\partial (\theta,\xi)$; see Section~\ref{subsec:stochastic_gradient} for our implementation. 
For the latter indifferentiable function $U_h(\xi)$, we define its subgradient $(0,v_h(\xi))$, where $v_h(\xi)$ is a member of the subdifferential set $\{v \mid \forall \xi',V_h(\xi') \ge V_h(\xi)+\innerprod{v}{\xi'-\xi}\}$. Namely, $-(0,v_h(\xi))$ is the supergradient of $-V_h$. Then, we define the stochastic gradient-supergradient: 
\[
    g_{h,\lambda}^{(t)}(\theta,\xi)
    =
    u_{\lambda}^{(t)}(\theta,\xi)
    -
    (0,v_h(\xi)). 
\]

With inspiration from the difference-of-convex algorithm \citep{tao1997convex}, we propose a stochastic gradient-supergradient descent (SGSD) algorithm that updates the parameters  $(\theta^{(t)},\xi^{(t)})$ at iteration $t$ by:
\[
    (\theta^{(t+1)},\xi^{(t+1)})
    \leftarrow 
    (\theta^{(t)},\xi^{(t)})
    -
    \omega_t 
    g_{\lambda}^{(t)}(\theta^{(t)},\xi^{(t)}).
\]
Since the proximal mapping of $T_h$ can be computed efficiently~\citep{yagishita2024fast}, one might consider applying stochastic proximal gradient methods. However, existing proximal methods such as those proposed by \citet{xu2019non} and \citet{yun2021adaptive} require the mini-batch size to grow indefinitely in order to guarantee convergence when dealing with integral-based regularization terms. This requirement makes them highly inefficient for our setting. In contrast, the proposed SGSD algorithm does not require increasing the batch size, as discussed in Section~\ref{subsec:convergence_analysis}.

\paragraph{SGSD in Practical Scenarios:} The stochastic gradient-supergradient $g_{h,\lambda}^{(t)}(\theta,\xi)$ reduces to the standard stochastic gradient of $F_{h,\lambda}(\theta,\xi)$ at differentiable points. Since differentiable points are present almost everywhere in the parameter space, the proposed SGSD, whose convergence is rigorously proven in Section~\ref{subsec:convergence_analysis}, is essentially equivalent to standard SGD when applied to ARTL in practical scenarios. Therefore, our analysis also can be regarded as a strong theoretical evidence for the SGD applied to robust regression using neural networks.

\subsection{A Practical Stochastic Gradient}
\label{subsec:stochastic_gradient}

To complete the definition of SGSD presented in Section~\ref{subsec:SGSD}, this section provides a practical construction of the stochastic gradient $u_{\lambda}^{(t)}(\theta^{(t)},\xi^{(t)})$, which serves as an unbiased estimate of the smooth function $U_{\lambda}(\theta^{(t)},\xi^{(t)})$. The stochastic gradient is employed instead of the deterministic gradient because the gradient of $U_{\lambda}(\theta,\xi)$ includes an integral term arising from the HOVR term~\eqref{eq:HOV}. 

We begin by specifying the hyperparamter $M^{(t)} \in \mathbb{N}$, and subsequently generate i.i.d. random numbers $\mathcal{Z}^{(t)}:= \{z_1^{(t)},z_2^{(t)},\ldots,z_{M^{(t)}}^{(t)}\}$ from a uniform distribution over the set $\Omega$. We then define:
\begin{align}
    &u_{\lambda}^{(t)}(\theta^{(t)},\xi^{(t)}) 
    =
    \frac{1}{n} \frac{\partial \{\|r(\theta^{(t)})-\xi^{(t)}\|_2^2 + \|\xi^{(t)}\|_2^2\}}{\partial (\theta,\xi)}
    +
    \lambda
    \sum_{\bs i}
    w_{\bs i}
    \frac{1}{M^{(t)}}
    \sum_{z \in \mathcal{Z}^{(t)}}
    \frac{\partial |\nabla^{[k]}_{\bs i} f_{\theta}(z)|^q
    }{\partial (\theta,\xi)},
    \label{eq:out_stochastic_gradient}
\end{align} 
and it unbiasedly estimates the gradient of $U_{\lambda}$ as: 
\[
\mathbb{E}[u_{\lambda}^{(t)}(\theta^{(t)},\xi^{(t)})]=\frac{\partial U_{\lambda}(\theta^{(t)},\xi^{(t)})}{\partial (\theta,\xi)}
\]
under the assumption that the order of operations can be exchanged. 
Given the general difficulty of deriving an explicit form of HOVR and the computational infeasibility of performing numerical integration at each step of gradient descent, this stochastic gradient is defined to bypass the need for integration.

\subsection{Convergence Analysis}

\label{subsec:convergence_analysis}

This section presents the convergence properties of SGSD. 
The following assumptions form the basis of our analysis: let ${\mathcal{F}_t}$ denote the natural filtration associated with ${(\theta^{(t)}, \xi^{(t)})}$, and let $\mathbb{E}_t$ represent the conditional expectation, $\mathbb{E}[ \, \cdot \mid \mathcal{F}_t]$.

\begin{assumption}\label{assume:SGSD}
    We assume the following conditions:
    \begin{enumerate}[(a)]
        \item $\partial U_{\lambda}/\partial (\theta,\xi)$ is Lipschitz continuous with a Lipschitz constant $L>0$; 
        \label{enum:Lipschitz}
        \item $\mathbb{E}_t[u_{\lambda}^{(t)}(\theta^{(t)},\xi^{(t)})] = \partial U_{\lambda}(\theta^{(t)},\xi^{(t)})/\partial (\theta,\xi)$;
        \label{enum:unbiasedness}
        \item There exist $\mu_1, \mu_2>0$ such that
        \begin{equation*}
            \mathbb{E}[\|g_{h,\lambda}^{(t)}(\theta^{(t)},\xi^{(t)})\|_2^2] \le \mu_1 + \mu_2\mathbb{E}[\|g_{h,\lambda}(\theta^{(t)},\xi^{(t)})\|_2^2],
        \end{equation*}
        where $g_{h,\lambda}(\theta^{(t)},\xi^{(t)})=\mathbbm{E}_t[g_{h,\lambda}^{(t)}(\theta^{(t)},\xi^{(t)})]$.
        \label{enum:stochastic_variation}
    \end{enumerate}
\end{assumption}

Assumption~\ref{assume:SGSD}(\ref{enum:Lipschitz}) is used commonly for first-order methods in optimization~\citep{bottou2018optimization}. 
Assumption~\ref{assume:SGSD}(\ref{enum:unbiasedness}) asserts the unbiasedness of the stochastic gradient. By defining $\Delta u_{\lambda}^{(t)} := u_{\lambda}^{(t)}(\theta^{(t)},\xi^{(t)})-\partial U_{\lambda}(\theta^{(t)},\xi^{(t)})/\partial (\theta,\xi)$, the boundedness condition $\mathbb{E}_t[\|\Delta u_{\lambda}^{(t)}\|_2^2] \le \sigma^2$, which is commonly used in the analysis of stochastic gradient descent, implies Assumption~\ref{assume:SGSD}(\ref{enum:stochastic_variation}), which is a natural extension of Assumption 4.3 in \citet{bottou2018optimization}. 
For a detailed derivation, refer to Supplement~\ref{app:derivation_of_SG_bound}.

Due to the strictness of the notation, we herein consider the convergence of a criticality measure
\begin{align*}
    \mathsf{C}(\theta,\xi)
    &:=\dist\left(\partial U_{\lambda}(\theta,\xi)/\partial (\theta,\xi),\{0\}\times\partial V_h(\xi)\right) \\
    &=\inf_{v\in\partial V_h(\xi)} \|\partial U_{\lambda}(\theta,\xi)/\partial (\theta,\xi)-(0,v)\|_2.
\end{align*}
With the unbiasedness assumption, roughly speaking, the convergence of the criticality measure leads to the convergence of $g_{h,\lambda}(\theta^{(t)},\xi^{(t)})$ up to the subdifferential difference. It also indicates the gradient-supergradient of ARTL converges to $0$. See Theorem~\ref{thm:bound-SGSD} for the convergence of the criticality measure. 

Since the current setting involves non-smooth and non-convex optimization, we cannot directly guarantee that the criticality measure converges to $0$ as in convex optimization. However, two forms of convergence can be established.
(i) The first corresponds to the so-called randomized SGD, where the iteration at which the algorithm terminates, denoted by $\tau_T$, is randomly selected after fixing the total number of iterations $T$. This approach has been adopted in prior works such as \citet{ghadimi2013stochastic,xu2019non}.
(ii) The second guarantees that the optimization path necessarily passes through at least one point where the criticality measure becomes zero.

Due to the inherent difficulty of non-convex optimization, neither of these results directly implies convergence at iteration $t = T$. However, in practice, the learning rate decays over time, which suggests that the criticality measure at $t = T$ is typically close to $0$.

\begin{theorem}\label{thm:bound-SGSD}
    \begin{enumerate}[(i)]
    Let $T$ be a natural number representing the number of SGSD iterations.
        \item 
        Suppose that Assumption \ref{assume:SGSD} holds and that the learning rate satisfies $0<\omega_s<2/(L\mu_2)$ for $s=0,1,...,T$.
        Then, we have
        \begin{align*}
            &\mathbb{E}\left[\mathsf{C}(\theta^{(\tau_{{}_T})},\xi^{(\tau_{{}_T})})^2 \right]  
            \le \frac{2F_{h,\lambda}(\theta^{(0)}, \xi^{(0)}) + L\mu_1\sum_{t=0}^T \omega_t^2}{\sum_{t=0}^T (2\omega_t-L\mu_2\omega_t^2)}, \nonumber
        \end{align*}
        where the stopping time $\tau_{{}_T} \in \{0,1,,...,T\}$ is randomly chosen independently of $\{(\theta^{(t)},\xi^{(t)})\}$, with the probability for $s=0,1,\ldots,T$:
        \[
        \mathbb{P}(\tau_{{}_T}=s)=\frac{2\omega_s-L\mu_2\omega_s^2}{\sum_{t=0}^{T}(2\omega_t-L\mu_2\omega_t^2)}.
        \]
        \item Suppose that Assumption \ref{assume:SGSD} holds and that the learning rate satisfies $0<\omega_s<2/(L\mu_2)$.
        Then, we have
        \begin{align*}
            &\min_{t=0,...,T}\mathbb{E}\left[\mathsf{C}(\theta^{(t)},\xi^{(t)})^2\right]  
            \le \frac{2F_{h,\lambda}(\theta^{(0)}, \xi^{(0)}) + L\mu_1\sum_{t=0}^T \omega_t^2}{\sum_{t=0}^T (2\omega_t-L\mu_2\omega_t^2)}. \nonumber
        \end{align*}
    \end{enumerate}
\end{theorem}

See Supplement~\ref{supp:proof_of_thm:bound-SGSD} for the proof.
Despite the non-smooth nature of our optimization problem, we establish convergence bounds that are comparable to those established in the smooth setting~\citep{ghadimi2013stochastic}. 
Considering the rate of convergence, we last discuss the learning rate schedule. 
If we use common diminishing learning rates that satisfy 
\begin{equation}\label{eq:stepsize-condition}
    \sum_{t=0}^\infty \omega_t=\infty, \, \text{ and } \, \sum_{t=0}^\infty \omega_t^2<\infty,
\end{equation}
the expected measures $\mathbb{E}\left[\mathsf{C}(\theta^{(\tau_{{}_T})},\xi^{(\tau_{{}_T})})^2 \right]$ at $\tau_T$ and $\min_{t=0,...,T}\mathbb{E}\left[\mathsf{C}(\theta^{(t)},\xi^{(t)})^2\right]$ converge to $0$ as $T \to \infty$.
Especially, since
\begin{align*}
    \mathbb{E}\left[\left(\min_{t=0,...,T}\mathsf{C}(\theta^{(t)},\xi^{(t)})\right)^2\right] \le \min_{t=0,...,T}\mathbb{E}\left[\mathsf{C}(\theta^{(t)},\xi^{(t)})^2\right],
\end{align*}
Theorem~\ref{thm:bound-SGSD} (ii) implies the $L^2$ convergence of $\min_{t=0,...,T}\mathsf{C}(\theta^{(t)},\xi^{(t)})$.
A typical choice, such as $\omega_t=\alpha t^{-1}$, yields a convergence rate of $\mathcal{O}((\log T)^{-1})$. However, there exist alternative learning rate schedules that achieve faster convergence, even though they violate the condition~\eqref{eq:stepsize-condition}. See Corollary~\ref{cor:convergence-rate-SGSD}.

\begin{corollary}\label{cor:convergence-rate-SGSD}
    Suppose that Assumption \ref{assume:SGSD} holds.
    Let $\omega_t=\alpha(1+t)^{-1/2}$ with $0<\alpha<2/(L\mu_2)$.
    Then, $\omega_t$ violates the condition \eqref{eq:stepsize-condition} while it attains
    \[
    \mathbb{E}\left[\mathsf{C}(\theta^{(\tau_{{}_T})},\xi^{(\tau_{{}_T})})^2 \right]
    =
    \mathcal{O}(T^{-1/2}\log{T}).
    \]
\end{corollary}

A similar result holds for $\min_{t=0,...,T}\mathbb{E}\left[\mathsf{C}(\theta^{(t)},\xi^{(t)})^2\right]$. 
Extending the Corollary \ref{cor:convergence-rate-SGSD}, we can also prove that $\beta=1/2$ yields the fastest convergence rate (in terms of the upper-bound) among the learning rate schedule $\omega_t=(1+t)^{-\beta}$, $\beta>0$. 
As it diverges from the main focus of this paper, it will be omitted.

\section{Experiments}
\label{sec:experiments}

The proposed approach is evaluated through synthetic dataset experiments in Section~\ref{subsec:synthetic_dataset_experiments}. We explore parameter selection via cross-validation, utilizing the robust trimmed loss in Section~\ref{subsec:robust_parameter_selection}. Benchmark dataset experiments are conducted in Section~\ref{subsec:benchmark_dataset_experiments}. 
Source codes to reproduce these experimental results are provided in \url{https://github.com/oknakfm/ARTL}.

\subsection{Synthetic Dataset Experiments}

\label{subsec:synthetic_dataset_experiments}

In this section, we assess the performance of our method on synthetic datasets. We compare it against several robust regression baselines to evaluate their effectiveness in handling outliers.

\paragraph{Dataset Generation:} We consider four different ground-truth functions $f(x)$ to generate the datasets:
\begin{itemize}
    \item checkered: $\sin(2x_1) \cos(2x_2)$,
    \item volcano: $\exp\left( -\left( \{x_1 - \pi\}^2 + \{x_2 - \pi\}^2 - 1 \right)^2 \right)$,
    \item stripe: $\sin\left( 2(x_1 + x_2) \right)$,
    \item plane: $x_1-x_2$. 
\end{itemize}

Refer to Figure~\ref{fig:true_functions} in Supplement~\ref{supp:experiments} for an illustration of these functions. 
For each function, we generate a grid of $n = 10^2 = 100$ data points over the domain $[0, 2\pi]^2$. The input variables $(x_1, x_2)$ are uniformly spaced within this domain. The target variable $y$ is computed by adding Gaussian noise to the true function value: $y = f(x) + \epsilon$, where $\epsilon \sim N(0, 0.2^2)$.

\paragraph{Outlier Injection:} To simulate outliers, we introduce anomalies into the dataset by replacing the target values of a small subset of data points. Specifically, we randomly select $3\%$ of the data points and assign their target values as $y_i = 5 + \delta_i$, where $\delta_i$ is a small random perturbation sampled uniformly from $[-0.1, 0.1]$.

\paragraph{Implementation Details:} 
All experiments are implemented using \verb|Python| with the \verb|NumPy|, \verb|PyTorch|, and \verb|scikit-learn| libraries. The neural networks consist of three hidden layers with 100 neurons each (i.e., MLP defined in Section \ref{subsec:NN} with $Q=3,L_1=L_2=L_3=100$) and use the sigmoid activation function. Training is performed using the Adam optimizer with a learning rate of 0.01 for 5000 iterations. As a learning rate schedule, step decay (with $\gamma=0.5$ for each $1000$ iteration) is employed. Unless otherwise noted, the HOVR regularization term uses $\lambda = 10^{-3}$, $q = 2$, and $w_{\bs i}=J^{-1}\mathbbm{1}(i_1=i_2=\cdots=i_k)$.

\paragraph{Baselines:} We employ the following baselines.

\begin{itemize}
    \item Linear regression with a robust Huber's loss~\citep{huber1981robust}.
    \item Linear regression with RANSAC~\citep{fischler1981random}: An iterative method that fits a model to subsets of the data to identify inliers, effectively excluding outliers from the final model.
    \item Support vector regression~\citep{drucker1997support}: A regression method using the $\epsilon$-insensitive loss function and an RBF kernel to capture nonlinear relationships.
    \item Neural network with robust Huber's Loss~\citep{huber1981robust} and Tukey's biweight loss~\citep{beaton1974fitting}. Tukey's loss is implemented using the constant $c=4.685$. 
    \item Neural network with a label noise regularization~\citep{han2020sigua}, that is designed to mitigate the impact of noisy labels.
    \item Neural Network with RANSAC-like Approach \citep{choi2009performance}: A neural network that iteratively excludes high-loss samples during training, similar to the RANSAC method~\citep{fischler1981random}.
\end{itemize}

Our proposed method, SGSD applied to TTL+HOVR, is evaluated with different two orders of derivatives $k \in \{1, 2\}$. $h$ is specified as a rounded integer of $0.9n$.

\paragraph{Evaluation Metrics:} We assess the performance of each method using the predictive mean squared error (PMSE) on a separate test set of $n_{\text{test}} = 10^4$ randomly sampled points from the same domain. Each experiment is repeated five times with different random seeds, and we report the mean and standard deviation of the PMSE across these runs.

\begin{table*}[!ht]
\centering
\caption{Synthetic dataset experiments. The proposed approach (NN with ARTL) is compared with robust regression methods. The \best{best score} is bolded and blue-colored, and the \second{second best score} is underlined.}
\label{table:synthetic-comparison}
\scalebox{0.85}{
\begin{tabular}{l|ccc|c}
\toprule 
& \multicolumn{3}{c|}{Non-linear} & Linear \\
& checkered & volcano & stripe & plane \\
\hline
Linear Reg. with Huber's Loss & $0.124 \, (0.004)$ & $0.130 \, (0.003)$ & $0.498 \, (0.016)$ & $\second{0.001} \, (0.001)$ \\ 
Linear Reg. with RANSAC & $0.140 \, (0.013)$ & $0.186 \, (0.058)$ & $0.871 \, (0.277)$ & $\best{0.001} \, (0.001)$\\
Support Vector Reg. with RBF Kernel & $0.127 \, (0.008)$ & $0.113 \, (0.020)$ & $0.508 \, (0.031)$ & $0.006 \, (0.002)$\\
NN with Huber's Loss & $0.634 \, (0.608)$ & $1.031 \, (0.861)$ & $0.488 \, (0.475)$ & $0.043 \, (0.025)$\\
NN with Tukey's Loss & $0.458 \, (0.655)$ & $0.413 \, (0.630)$ & $0.304 \, (0.395)$ & $0.017 \, (0.009)$\\
NN with Label Noise Reg. & $1.155 \, (1.068)$ & $0.872 \, (0.659)$ & $0.561 \, (0.498)$ & $0.756 \, (0.945)$\\
NN with RANSAC & $0.160 \, (0.036)$ & $0.142 \, (0.009)$ & $0.527 \, (0.018)$ & $0.011 \, (0.016)$\\
\rowcolor{gray!30}
NN with ARTL ($h=0.9n,k=1$) & $\second{0.088} \, (0.090)$ & $\second{0.082} \, (0.076)$ & $\second{0.223} \, (0.238)$ & $0.010 \, (0.004)$ \\
\rowcolor{gray!30}
NN with ARTL ($h=0.9n,k=2$) & $\best{0.061} \, (0.016)$ & $\best{0.040} \, (0.029)$ & $\best{0.119} \, (0.047)$ & $0.007 \, (0.001)$ \\
\bottomrule
\end{tabular}
}
\end{table*}

\paragraph{Results:} Experimental results are summarized in Table~\ref{table:synthetic-comparison}. Visualization of the estimated functions is also listed in Supplement~\ref{supp:experiments}. Our interpretation is summarized as follows. 

For the linear `plane' function, which can be effectively modeled by simple linear approaches, linear regression using robust loss functions achieves excellent prediction accuracy. While simpler models like linear regression and support vector regression are well-suited for predicting linear functions, the proposed approach using neural network at least achieves the best performance among neural network-based methods.

In contrast, for the non-linear `checkered', `volcano', and `stripe' functions, simpler models are unable to capture the complex patterns. The proposed neural network with TTL+HOVR achieves the best scores across all methods. It is important to note that neural networks trained with robust loss functions, such as Huber's and Tukey's, tend to overfit to outliers, as illustrated in Figure~\ref{fig:illustration}.

\paragraph{Ablation Study:} 

To demonstrate the effectiveness of combining trimmed loss with HOVR, we present an ablation study in Figure~\ref{fig:ablation} with the `checkered' function. 
As we can easily examine, (\subref{subfig:trim}): The NN trained with the trimmed loss (without HOVR) reduces the overall influence of the outliers due to the robust loss but still overfits to the outliers themselves. As a result, the non-outlier regions are correctly estimated, but sharp peaks appear near the outliers. (\subref{subfig:HOVR}): The influence of outliers is not properly handled by the NN trained only with HOVR (without the robust loss). The regularization suppresses the overall variability of the neural network, but the presence of outliers affects not only the regions with outliers but also other areas of the entire covariate space. (\subref{subfig:trim+HOVR}): The proposed ARTL (TTL+HOVR) yields significantly improved results.

\begin{figure}[!ht]
\centering
\begin{minipage}{0.24\textwidth}
\centering
\includegraphics[width=0.9\textwidth]{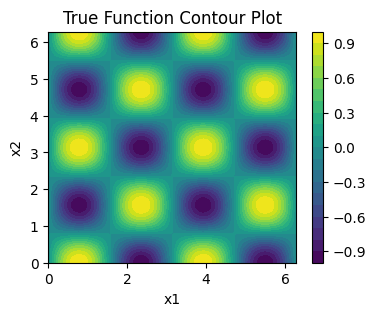}
\subcaption{True}
\end{minipage}
\begin{minipage}{0.24\textwidth}
\centering
\includegraphics[width=0.9\textwidth]{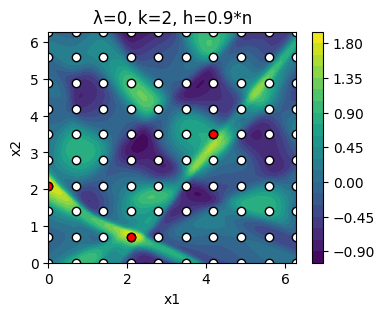}
\subcaption{Only TTL}
\label{subfig:trim}
\end{minipage} 
\begin{minipage}{0.24\textwidth}
\centering
\includegraphics[width=0.9\textwidth]{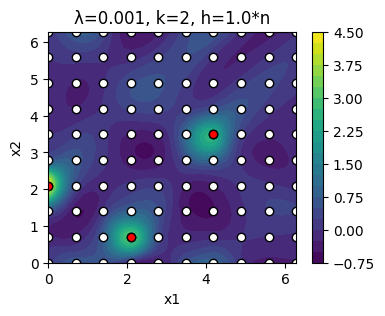}
\subcaption{Only HOVR}
\label{subfig:HOVR}
\end{minipage}
\begin{minipage}{0.24\textwidth}
\centering
\includegraphics[width=0.9\textwidth]{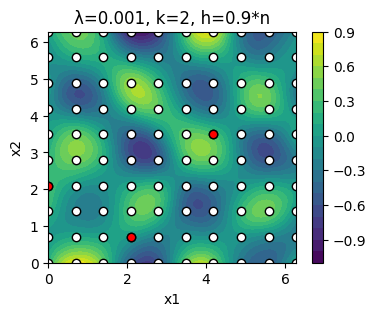}
\subcaption{TTL+HOVR}
\label{subfig:trim+HOVR}
\end{minipage}
\caption{Ablation study.}
\label{fig:ablation}
\end{figure}

The average PMSE values (and the standard deviations) are: (\subref{subfig:trim}) $0.187 \, (0.021)$, (\subref{subfig:HOVR}) $0.356 \, (0.057)$, (\subref{subfig:trim+HOVR}) $0.064 \, (0.005)$, and these values clearly highlight the superior performance of our proposed approach (ARTL).



\begin{table*}[!ht]
\centering 
\caption{Benchmark dataset experiments. The proposed approach (NN with ARTL) is compared with robustly trained neural networks. The \best{best score} is bolded and blue-colored, and the \second{second best score} is underlined.}
\label{table:benchmark}
\scalebox{0.85}{
\begin{tabular}{lccc}
\toprule
& Auto MPG & Liver Disorders & Real Estate Valuation \\
\hline
NN with Huber's Loss & $0.808 \, (0.363)$ & $3.269 \, (0.798)$ & $2.148 \, (0.710)$ \\
NN with Tukey's Loss & $1.126 \, (0.839)$ & $2.976 \, (0.477)$ & $2.127 \, (0.746)$ \\
NN with Label Noise Reg. & $0.751 \, (0.410)$ & $2.016 \, (0.319)$ & $1.433 \, (0.383)$ \\
NN with RANSAC & $1.078 \, (0.172)$ & $\best{1.146} \, (0.158)$ & $1.087 \, (0.158)$ \\
\rowcolor{gray!30}
NN with ARTL ($h=0.9n,k=1,\lambda=10^{-4}$) & $0.229 \, (0.040)$ & $1.617 \, (0.273)$ & $0.470 \, (0.132)$ \\
\rowcolor{gray!30}
NN with ARTL ($h=0.9n,k=1,\lambda=10^{-3}$) & $\best{0.211} \, (0.038)$ & $\second{1.546} \, (0.190)$ & $\best{0.453} \, (0.144)$ \\
\rowcolor{gray!30}
NN with ARTL ($h=0.9n,k=2,\lambda=10^{-4}$) & $0.228 \, (0.053)$ & $1.554 \, (0.295)$ & $\second{0.467} \, (0.166)$ \\
\rowcolor{gray!30}
NN with ARTL ($h=0.9n,k=2,\lambda=10^{-3}$) & $\second{0.223} \, (0.064)$ & $1.633 \, (0.145)$ & $0.479 \, (0.155)$ \\
\bottomrule
\end{tabular}
}
\end{table*}

\subsection{Robust Parameter Selection}
\label{subsec:robust_parameter_selection}

In this section, we discuss the process of robust parameter selection within our proposed framework. The robust validation score is defined as follows:

\paragraph{Robust Validation Score:} We start by splitting the outlier-contaminated data into a training set (80\%) and a validation set (20\%). The validation score is computed using the robust trimmed loss ($h=0.9n$).

In this experiment, we compute the average PMSE over $10^4$ randomly generated test samples (excluding outliers) and the average validation score (using the outlier-contaminated data) across $10$ trials, each with a different random seed. The majority of the experimental settings are inherited from Section~\ref{subsec:synthetic_dataset_experiments}.

\paragraph{Results:} Figure~\ref{fig:validation} presents the results. The averaged PMSE and the averaged validation scores are displayed, with error bars indicating the standard deviation. The strong correlation highlights that the robust validation score, based on the trimmed loss, is a reliable estimator of PMSE. Consequently, we think that this robust validation score can be effectively used for parameter selection, even in the presence of outliers.

\begin{figure}[!ht]
\centering
\includegraphics[width=0.4\textwidth]{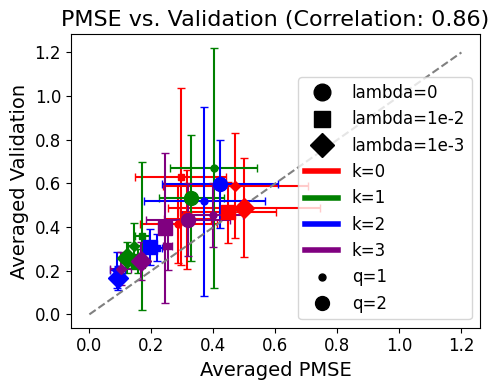}
\caption{The average robust validation scores are compared with the average PMSE, in the presence of outliers. The Pearson correlation coefficient between the averaged values is $0.86$, and the Spearman's rank correlation coefficient is $0.89$.}
\label{fig:validation}
\end{figure}

\subsection{Benchmark Dataset Experiments}
\label{subsec:benchmark_dataset_experiments}

We validate our proposed algorithm through experiments on typical benchmark datasets: the Auto MPG, Liver Disorders, and Real Estate Valuation datasets from the UCI Machine Learning Repository~\citep{UCI}\footnote{\url{https://archive.ics.uci.edu}}, with artificially introduced outliers. For the Auto MPG dataset, the `origin' and `car\_name' columns were removed, and for the Real Estate Valuation dataset, the `X1 transaction date', `X5 latitude', and `X6 longitude' columns were excluded due to a data type issue. Consequently, the sample size and the covariate dimensions are, $n=398, J=5$ for Auto MPG, $n=345, J=5$ for Liver Disorders, and $n=414, J=3$ for Real Estate Valuation datasets. Implementation details are inherited from Section~\ref{subsec:synthetic_dataset_experiments}.

\paragraph{Dataset Formatting:} In each experiment among $5$ runs, we randomly split each dataset into a training set (70\%) and a test set (30\%). To introduce outliers, 5\% of the outcomes in the training set are randomly selected and modified by adding two times the standard deviation to their original values. 

\paragraph{Baselines:} We utilize non-linear neural networks with Huber’s and Tukey’s loss functions, label noise regularization, and RANSAC, as also employed in synthetic dataset experiments in Section~\ref{subsec:synthetic_dataset_experiments}.

\paragraph{Results:} Experimental results are summarized in Table~\ref{table:benchmark}. 
Average PMSE among the $5$ runs is provided with the standard deviation in the parenthesis. 
For all datasets, the proposed method consistently achieves the best or nearly best results. While the RANSAC approach yields the top score for the Liver Disorders dataset, its performance is unstable, and it produces the second worst score for the Auto MPG dataset.

\section{Conclusion}
\label{sec:conclusion}

In this study, we introduced a stochastic optimization algorithm called stochastic gradient–supergradient descent (SGSD) for training neural networks using the augmented and regularized trimmed loss (ARTL). ARTL combines a transformed trimmed loss (TTL) with higher-order variation regularization (HOVR). We proved that minimizing ARTL yields a high functional breakdown point, defined in terms of HOV. Furthermore, we provided a theoretical proof of the convergence of SGSD and demonstrated its effectiveness through experiments. A key advantage of our method, compared to traditional techniques such as splines, is its broad applicability to a wide range of models, including neural networks with complex architectures.

\section*{Acknowledgement}
A. Okuno was supported by JSPS KAKENHI (21K17718, 22H05106, 25K03087). 
S. Yagishita was supported by JSPS KAKENHI (25K21158). 
We would like to thank Hironori Fujisawa and Keisuke Yano for the helpful discussions.


\clearpage

\appendix

\onecolumn

{\LARGE\bf Supplementary Materials}

\begin{itemize}
\item (Title) Outlier-robust neural network training: variation regularization meets trimmed loss to prevent functional breakdown
\item (Authors) Akifumi Okuno and Shotaro Yagishita
\end{itemize}

\section{Regularization for Simpler Models}
\label{app:regularization_kernel}
A linear basis function model $f^{\dagger}_{\theta}(x)=\sum_{l=1}^{L} \theta_l \phi_l(x)$ with user-specified basis functions $\{\phi_l\}$ encompasses kernel and spline regression models as special cases. 
Here, consider the univariate case $J=1$ for simplicity. 
With the gram matrix $G^{[k]}=(g^{[k]}_{ij}),g_{ij}^{[k]}=\int_{\Omega}\nabla^{[k]}\phi_i(x) \nabla^{[k]}\phi_j(x)\diff x$, we have 
\begin{align*}
    \|\nabla^{[k]} f^{\dagger}_{\theta}\|_{L^2(\Omega)}^2
    &=
    \int_{\Omega} \{\nabla^{[k]}f^{\dagger}_{\theta}(x)\}^2 \diff x \\
    &=
    \sum_{l_1=1}^{L}\sum_{l_2=1}^{L}
    \theta_{l_1} \theta_{l_2} 
    \underbrace{\int_{\Omega} \nabla^{[k]}\phi_{l_1}(x)\nabla^{[k]}\phi_{l_2}(x) \diff x}_{=g_{l_1l_2}} \\
    &=
    \sum_{l_1=1}^{L}\sum_{l_2=1}^{L}\theta_{l_1}\theta_{l_2} g_{l_1l_2}^{[k]}
    =
    \langle \theta, G^{[k]} \theta \rangle.
\end{align*}
It also reduces to the ridge regularization $\|\theta\|_2^2$ by assuming that $\{\nabla^{[k]}\phi_l\}$ is orthonormal (i.e., $G^{[k]}$ is an identity matrix). See, e.g., \citet{smale2007learning} for more details of the regularization for kernel regression.

\section{Detailed Proofs}
\label{supp:proofs}

\subsection{Proof of Theorem~\ref{theo:breakdown_point}}
\label{supp:proof_of_theo:breakdnwon_point}

Fix \( \theta_0 \in \Theta \) arbitrarily. 
Consider an arbitrary contaminated dataset \( \widetilde{Z}^{(m)} \) satisfying \( m \le n-h \).  
Let \( I_R^{(m)} \) denote the index set of the data points that have not been replaced in \( \widetilde{Z}^{(m)} \).  
We define the loss function and the residuals determined by \( \widetilde{Z}^{(m)} \) as \( \widetilde{L}^{(m)}(\theta) \) and \( \widetilde{r}^{(m)}(\theta) \), respectively.  
Letting \( \widetilde{Z}^{(m)} = \{(\Tilde{x}_i, \Tilde{y}_i)\}_{i=1}^{n} \),
the following holds: 
\begin{align*}
    \widetilde{L}^{(m)}(\theta_0) 
    &=
    \frac{1}{2} T_h(\widetilde{r}^{(m)}(\theta_0)) + \lambda C_{k,q}(f_{\theta_0}) \\
    &= 
    \min_{I_h \subset \{1,2,\ldots,n\}:|I_h|=h}
    \frac{1}{2n}
    \sum_{i \in I_h}
    (\tilde{y}_i-f_{\theta_0}(\tilde{x}_i))^2
    +
    \lambda C_{k,q}(f_{\theta_0}) \\
    &\le 
    \frac{1}{2n} \sum_{i \in I_R^{(m)}} (\Tilde{y}_i - f_{\theta_0}(\Tilde{x}_i))^2 + \lambda C_{k,q}(f_{\theta_0}) 
    \qquad (\because \, |I_R^{(m)}| \ge h) \\
    &=
    \frac{1}{2n} \sum_{i \in I_R^{(m)}} (y_i - f_{\theta_0}(x_i))^2 + \lambda C_{k,q}(f_{\theta_0}) 
    \qquad (\because \, (\Tilde{x}_i, \Tilde{y}_i) = (x_i, y_i) ~~\mbox{for}~~ i \in I_R^{(m)}) \\
    &\le 
    \frac{1}{2n} \sum_{i=1}^n (y_i - f_{\theta_0}(x_i))^2 + \lambda C_{k,q}(f_{\theta_0}) 
    =: M_0.
\end{align*}
Note that \( M_0 \) is finite and independent of \( \widetilde{Z}^{(m)} \).
Letting \( \displaystyle \hat{\theta}^{(m)} \in \argmin_{\theta \in \Theta} \widetilde{L}^{(m)}(\theta) \), we have:
\begin{equation}
    C_{k,q}(f_{\hat{\theta}^{(m)}}) 
    \le \frac{\widetilde{L}^{(m)}(\hat{\theta}^{(m)})}{\lambda} 
    \le \frac{\widetilde{L}^{(m)}(\theta_0)}{\lambda} 
    \le \frac{M_0}{\lambda}
    < \infty,
\end{equation}
which indicates the assertion.

\qed

\subsection{Derivation of Assumption~\ref{assume:SGSD}(\ref{enum:stochastic_variation})}
\label{app:derivation_of_SG_bound}

As it holds that
\begin{align}
    \mathbb{E}[\innerprod{\Delta u_{h,\lambda}^{(t)}}{g_{h,\lambda}(\theta^{(t)},\xi^{(t)})}]
    = \mathbb{E}[\mathbb{E}_t[\innerprod{\Delta u_{h,\lambda}^{(t)}}{g_{h,\lambda}(\theta^{(t)},\xi^{(t)})}]] 
    = \mathbb{E}[\innerprod{\underbrace{\mathbb{E}_t[\Delta u_{h,\lambda}^{(t)}]}_{=0}}{g_{h,\lambda}(\theta^{(t)},\xi^{(t)})}] = 0, \label{eq:ip0}
\end{align}
we have
\begin{align*}
    \mathbb{E}[\|g_{h,\lambda}^{(t)}(\theta^{(t)},\xi^{(t)})\|_2^2]
    &= 
    \mathbb{E}[\|\Delta u_{h,\lambda}^{(t)}+g_{h,\lambda}(\theta^{(t)},\xi^{(t)})\|_2^2] \\
    &= \mathbb{E}[\|\Delta u_{h,\lambda}^{(t)}\|_2^2] 
    + 2 \underbrace{\mathbb{E}[\innerprod{\Delta u_{h,\lambda}^{(t)}}{g_{h,\lambda}(\theta^{(t)},\xi^{(t)})}]}_{=0 \, (\because \, \text{equation} \, \eqref{eq:ip0})}
    + \mathbb{E}[\|g_{h,\lambda}(\theta^{(t)},\xi^{(t)})\|_2^2]\\
    &\le \sigma^2 + \mathbb{E}[\|g_{h,\lambda}(\theta^{(t)},\xi^{(t)})\|_2^2].
\end{align*}

\qed

\subsection{Proof of Theorem \ref{thm:bound-SGSD}}
\label{supp:proof_of_thm:bound-SGSD}
To prove Theorem \ref{thm:bound-SGSD}, we present the following lemma.

\begin{lemma}\label{lemma:bound-SGSD}
    Suppose that Assumption \ref{assume:SGSD} holds.
    Then, it holds that
    \begin{align}
        \sum_{t=0}^T (2\omega_t-L\mu_2\omega_t^2)\mathbb{E}\left[\mathsf{C}(\theta^{(t)},\xi^{(t)})^2\right] 
        \le 
        2F_{h,\lambda}(\theta^{(0)}, \xi^{(0)}) + L\mu_1\sum_{t=0}^T \omega_t^2
        \label{eq:lemma:bound-SGSD}
    \end{align}
    for any $T\ge0$.
\end{lemma}

\begin{proof}
    As the inequality
    \begin{align*}
        \mathsf{C}(\theta^{(t)},\xi^{(t)}) &= \inf_{v\in\partial V_h(\xi^{(t)})} \|\partial U_{\lambda}(\theta^{(t)},\xi^{(t)})/\partial (\theta,\xi)-(0,v)\|_2\\
        &\le \|\partial U_{\lambda}(\theta^{(t)},\xi^{(t)})/\partial (\theta,\xi)-(0,v_h(\xi^{(t)}))\|_2\\
        &= \|g_{h,\lambda}(\theta^{(t)},\xi^{(t)})\|_2
    \end{align*}
    holds, it suffices to prove that $\sum_{t=0}^{T}(2\omega_t-L\mu_2\omega_t^2)\mathbb{E}[\|g_{h,\lambda}(\theta^{(t)},\xi^{(t)})\|_2^2]$ (which upper-bounds the left-hand side of \eqref{eq:lemma:bound-SGSD}) is upper-bounded by the right-hand side of \eqref{eq:lemma:bound-SGSD}. 

    Since $\partial U_{\lambda}/\partial (\theta,\xi)$ is Lipschitz continuous with the Lipschitz constant $L$, it follows from the well-known descent lemma \citep[Lemma 5.7]{beck2017first} that
    \begin{align*}
        &U_{\lambda}(\theta^{(t+1)},\xi^{(t+1)})\\
        &\le U_{\lambda}(\theta^{(t)}, \xi^{(t)}) + \innerprod{\partial U_{\lambda}(\theta^{(t)},\xi^{(t)})/\partial (\theta,\xi)}{(\theta^{(t+1)},\xi^{(t+1)})-(\theta^{(t)},\xi^{(t)})} + \frac{L}{2}\|(\theta^{(t+1)},\xi^{(t+1)})-(\theta^{(t)},\xi^{(t)})\|_2^2.
    \end{align*}
    On the other hand, the definition of the subgradient $v_h$ implies that
    \begin{equation*}
        -V_h(\xi^{(t+1)}) \le -V_h(\xi^{(t)}) - \innerprod{v_h(\xi^{(t)})}{\xi^{(t+1)}-\xi^{(t)}}.
    \end{equation*}
    Consequently, we have
    \begin{align*}
        &F_{h,\lambda}(\theta^{(t+1)}, \xi^{(t+1)})\\
        &\le F_{h,\lambda}(\theta^{(t)}, \xi^{(t)}) + \innerprod{g_{h,\lambda}(\theta^{(t)},\xi^{(t)})}{(\theta^{(t+1)}, \xi^{(t+1)})-(\theta^{(t)}, \xi^{(t)})} + \frac{L}{2}\|(\theta^{(t+1)}, \xi^{(t+1)})-(\theta^{(t)}, \xi^{(t)})\|_2^2\\
        &=  F_{h,\lambda}(\theta^{(t)}, \xi^{(t)}) - \omega_t \innerprod{g_{h,\lambda}(\theta^{(t)},\xi^{(t)})}{g_{h,\lambda}^{(t)}(\theta^{(t)},\xi^{(t)})} + \frac{L\omega_t^2}{2}\|g_{h,\lambda}^{(t)}(\theta^{(t)},\xi^{(t)})\|_2^2.
    \end{align*}
    Taking the expectation yields
    \begin{align*}
        &\hspace{-2em}\mathbb{E}[F_{h,\lambda}(\theta^{(t+1)}, \xi^{(t+1)})]\\
        &\le \mathbb{E}[F_{h,\lambda}(\theta^{(t)}, \xi^{(t)})] - \omega_t \mathbb{E}[\innerprod{g_{h,\lambda}(\theta^{(t)},\xi^{(t)})}{g_{h,\lambda}^{(t)}(\theta^{(t)},\xi^{(t)})}] + \frac{L\omega_t^2}{2}\mathbb{E}[\|g_{h,\lambda}^{(t)}(\theta^{(t)},\xi^{(t)})\|_2^2]\\
        &\overset{\text{Tower property}}{=} \mathbb{E}[F_{h,\lambda}(\theta^{(t)}, \xi^{(t)})] - \omega_t \mathbb{E}[\mathbb{E}_t[\innerprod{g_{h,\lambda}(\theta^{(t)},\xi^{(t)})}{g_{h,\lambda}^{(t)}(\theta^{(t)},\xi^{(t)})}]] + \frac{L\omega_t^2}{2}\mathbb{E}[\|g_{h,\lambda}^{(t)}(\theta^{(t)},\xi^{(t)})\|_2^2]\\
        &= \mathbb{E}[F_{h,\lambda}(\theta^{(t)}, \xi^{(t)})] - \omega_t \mathbb{E}[\innerprod{g_{h,\lambda}(\theta^{(t)},\xi^{(t)})}{\mathbb{E}_t[g_{h,\lambda}^{(t)}(\theta^{(t)},\xi^{(t)})]}] + \frac{L\omega_t^2}{2}\mathbb{E}[\|g_{h,\lambda}^{(t)}(\theta^{(t)},\xi^{(t)})\|_2^2]\\
        &\overset{\text{Assumption}~ \ref{assume:SGSD}(\ref{enum:unbiasedness})}{=} \mathbb{E}[F_{h,\lambda}(\theta^{(t)}, \xi^{(t)})] - \omega_t \mathbb{E}[\|g_{h,\lambda}(\theta^{(t)},\xi^{(t)})\|_2^2] + \frac{L\omega_t^2}{2}\mathbb{E}[\|g_{h,\lambda}^{(t)}(\theta^{(t)},\xi^{(t)})\|_2^2]\\
        &\overset{\text{Assumption}~ \ref{assume:SGSD}(\ref{enum:stochastic_variation})}{\le} \mathbb{E}[F_{h,\lambda}(\theta^{(t)}, \xi^{(t)})] - \left(\omega_t-\frac{L\mu_2\omega_t^2}{2}\right) \mathbb{E}[\|g_{h,\lambda}(\theta^{(t)},\xi^{(t)})\|_2^2] + \frac{L\mu_1\omega_t^2}{2}.
    \end{align*}  
    By summing up, we obtain 
    \begin{align*}
        \sum_{t=0}^T \left(\omega_t-\frac{L\mu_2\omega_t^2}{2}\right) \mathbb{E}[\|g_{h,\lambda}(\theta^{(t)},\xi^{(t)})\|_2^2] &\le F_{h,\lambda}(\theta^{(0)}, \xi^{(0)}) \underbrace{- \mathbb{E}[F_{h,\lambda}(\theta^{(T+1)}, \xi^{(T+1)})]}_{\le0 \quad (\because~ \text{$F_{h,\lambda}\ge0$})} + \frac{L\mu_1}{2} \sum_{t=0}^T \omega_t^2\\
        &\le F_{h,\lambda}(\theta^{(0)}, \xi^{(0)}) + \frac{L\mu_1}{2} \sum_{t=0}^T \omega_t^2.
    \end{align*}
    The assertion is proved.
\end{proof}

\begin{proof}[Proof of Theorem \ref{thm:bound-SGSD} (i)]
    Lemma \ref{lemma:bound-SGSD} yields
    \begin{align*}
        \mathbb{E}\left[\mathsf{C}(\theta^{(\tau_{{}_T})},\xi^{(\tau_{{}_T})})^2\right]
        &= \sum_{s=0}^T\mathbb{E}\left[\mathsf{C}(\theta^{(\tau_{{}_T})},\xi^{(\tau_{{}_T})})^2~\middle|~\tau_{{}_T}=s\right]\mathbb{P}(\tau_{{}_T}=s) \nonumber \\
        &= \sum_{s=0}^T\mathbb{E}\left[\mathsf{C}(\theta^{(s)},\xi^{(s)})^2\right]\frac{2\omega_s-L\mu_2\omega_s^2}{\sum_{t=0}^{T}(2\omega_t-L\mu_2\omega_t^2)} \nonumber \\
        &= 
        \underbrace{
        \left\{ 
        \sum_{s=0}^T
        (2\omega_s-L\mu_2\omega_s^2)
        \mathbb{E}\left[\mathsf{C}(\theta^{(s)},\xi^{(s)})^2\right]
        \right\}
        }_{\le 2F_{h,\lambda}(\theta^{(0)},\xi^{(0)}) + L\mu_1\sum_{t=0}^{t}\omega_t^2
        \quad (\because~ \text{Lemma}~\ref{lemma:bound-SGSD})  
        }
        \frac{1}{\sum_{t=0}^{T}(2\omega_t-L\mu_2\omega_t^2)} \nonumber \\
        &\le \frac{2F_{h,\lambda}(\theta^{(0)}, \xi^{(0)}) + L\mu_1\sum_{t=0}^T \omega_t^2}{\sum_{t=0}^T (2\omega_t-L\mu_2\omega_t^2)}.
    \end{align*}
\end{proof}

\begin{proof}[Proof of Theorem \ref{thm:bound-SGSD} (ii)]
    Lemma \ref{lemma:bound-SGSD} together with the assumption $0<\omega_t<2/(L\mu_2)$ yields
    \begin{align*}
        \sum_{t=0}^T (2\omega_t-L\mu_2\omega_t^2) \left(\min_{t=0,...,T}\mathbb{E}\left[\mathsf{C}(\theta^{(t)},\xi^{(t)})^2\right]\right)
        &\le \sum_{t=0}^T (2\omega_t-L\mu_2\omega_t^2)\mathbb{E}\left[\mathsf{C}(\theta^{(t)},\xi^{(t)})^2\right] \nonumber \\
        &\le 2F_{h,\lambda}(\theta^{(0)}, \xi^{(0)}) + L\mu_1\sum_{t=0}^T \omega_t^2.
    \end{align*}
\end{proof}

\subsection{Proof of Corollary \ref{cor:convergence-rate-SGSD}}

\begin{proof}
    As a simple formula indicates
    \begin{align*}
        \sum_{t=0}^T(1+t)^{-1} &\le 1 + \int_0^T (1+t)^{-1} dt = 1 + \log(T+1), \text{ and}\\
        \sum_{t=0}^T(1+t)^{-1/2} &\ge \int_0^{T+1} (1+t)^{-1/2} dt = 2(T+2)^{1/2} - 2,
    \end{align*}
    Theorem \ref{thm:bound-SGSD} yields the assertion
    \begin{align*}
        \mathbb{E}\left[\mathsf{C}(\theta^{(\tau_{{}_T})},\xi^{(\tau_{{}_T})})^2\right] &\le \frac{2F_{h,\lambda}(\theta^{(0)},\xi^{(0)}) + L\mu_1\alpha^2\sum_{t=0}^T (1+t)^{-1}}{2\alpha\sum_{t=0}^T(1+t)^{-1/2} - L\mu_2\alpha^2\sum_{t=0}^T (1+t)^{-1}}\\
        &\le \frac{2F_{h,\lambda}(\theta^{(0)},\xi^{(0)}) + L\mu_1\alpha^2 + L\mu_1\alpha^2\log(T+1)}{4\alpha(T+2)^{1/2} - L\mu_2\alpha^2\log(T+1) - 4\alpha - L\mu_2\alpha^2}\\
        &= \mathcal{O}(T^{-1/2}\log{T}).
    \end{align*}
\end{proof}

\section{Additional Visualizations}
\label{supp:experiments}

This section provides additional visualizations for the synthetic dataset experiments described in Section~\ref{subsec:synthetic_dataset_experiments}. 
The four true functions utilized in this study are as follows:

\begin{itemize}
    \item (checkered) $f_1(x) = \sin(2x_1) \cos(2x_2)$
    \item (volcano) $f_2(x) = \scalebox{0.9}{$\exp\left( -\left( \{x_1 - \pi\}^2 + \{x_2 - \pi\}^2 - 1 \right)^2 \right)$}$
    \item (stripe) $f_3(x) = \sin\left( 2(x_1 + x_2) \right)$
    \item (plane) $f_4(x)=x_1-x_2$. 
\end{itemize}

These functions, visualized in Figure~\ref{fig:true_functions}, serve as the ground truth. The corresponding regression models estimated from outlier-contaminated datasets derived from these functions are presented in Figures~\ref{fig:supp_checkered}--\ref{fig:supp_plane}.

\begin{figure*}[!ht]
\centering
\begin{minipage}{0.23\textwidth}
\centering
\includegraphics[width=0.9\textwidth]{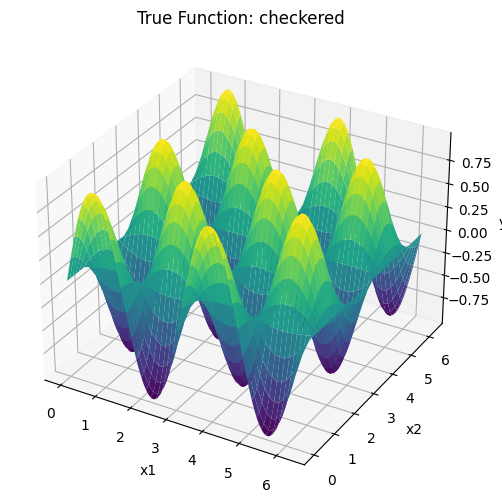}
\subcaption{$f_1$: checkered}
\label{subfig:f1}
\end{minipage}
\begin{minipage}{0.23\textwidth}
\centering
\includegraphics[width=0.9\textwidth]{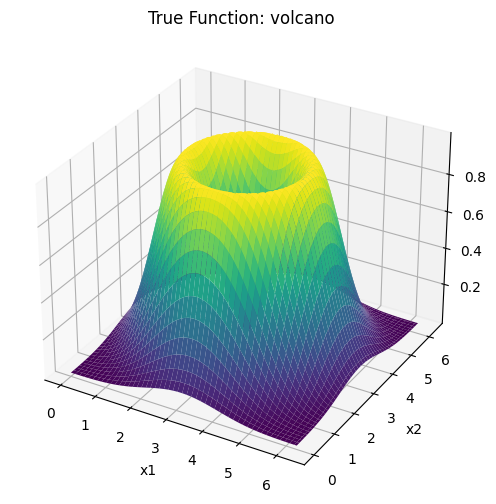}
\subcaption{$f_2$: volcano}
\label{subfig:f2}
\end{minipage} 
\begin{minipage}{0.23\textwidth}
\centering
\includegraphics[width=0.9\textwidth]{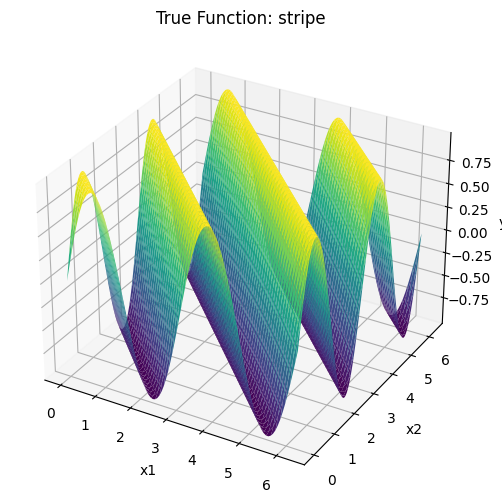}
\subcaption{$f_3$: stripe}
\label{subfig:f3}
\end{minipage}
\begin{minipage}{0.23\textwidth}
\centering
\includegraphics[width=0.9\textwidth]{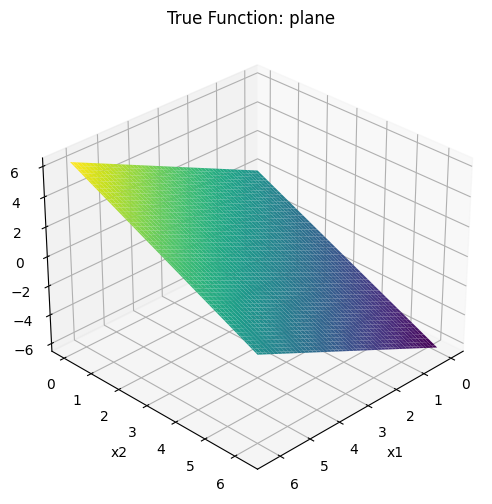}
\subcaption{$f_4$: plane}
\label{subfig:f4}
\end{minipage}
\caption{Underlying true functions in our synthetic dataset experiments.}
\label{fig:true_functions}
\end{figure*}

\begin{figure}[!p]
\centering
\fbox{
\begin{minipage}{0.32\textwidth}
\centering
\includegraphics[width=\textwidth]{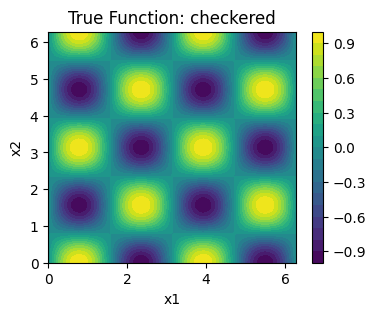}
\subcaption*{Underlying true function}
\end{minipage}
} \\
\begin{minipage}{0.32\textwidth}
\centering
\includegraphics[width=\textwidth]{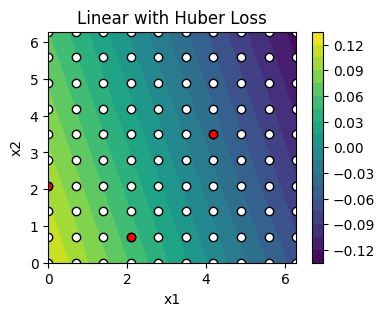}
\end{minipage}
\begin{minipage}{0.32\textwidth}
\centering
\includegraphics[width=\textwidth]{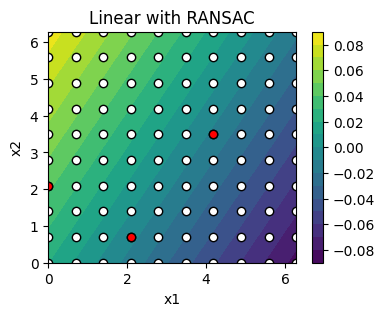}
\end{minipage}
\begin{minipage}{0.32\textwidth}
\centering
\includegraphics[width=\textwidth]{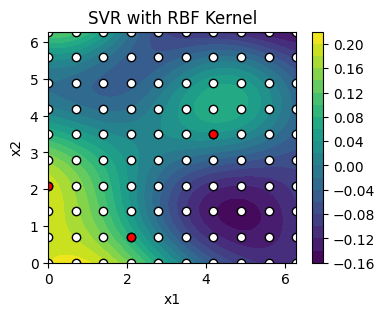}
\end{minipage}
\begin{minipage}{0.32\textwidth}
\centering
\includegraphics[width=\textwidth]{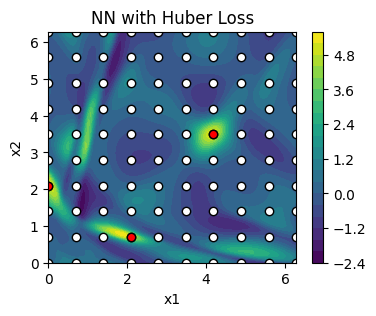}
\end{minipage}
\begin{minipage}{0.32\textwidth}
\centering
\includegraphics[width=\textwidth]{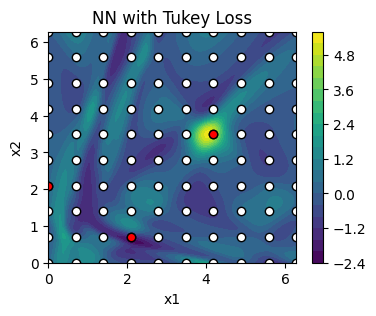}
\end{minipage}
\begin{minipage}{0.32\textwidth}
\centering
\includegraphics[width=\textwidth]{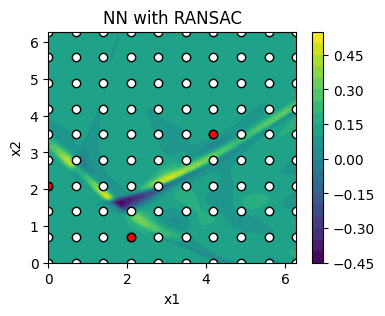}
\end{minipage}
\begin{minipage}{0.32\textwidth}
\centering
\includegraphics[width=\textwidth]{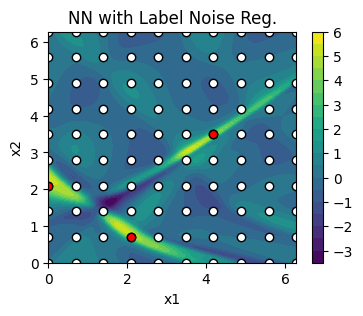}
\end{minipage}
\begin{minipage}{0.32\textwidth}
\centering
\includegraphics[width=\textwidth]{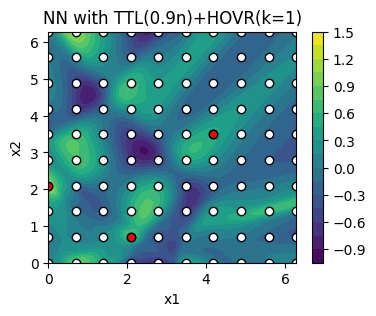}
\end{minipage}
\begin{minipage}{0.32\textwidth}
\centering
\includegraphics[width=\textwidth]{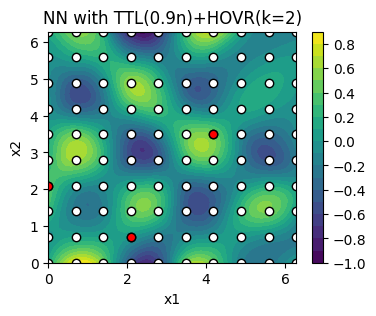}
\end{minipage}
\caption{Checkered function $f_1(x)=\sin(2x_1)\cos(2x_2)$.}
\label{fig:supp_checkered}
\end{figure}

\begin{figure}[!p]
\centering
\fbox{
\begin{minipage}{0.32\textwidth}
\centering
\includegraphics[width=\textwidth]{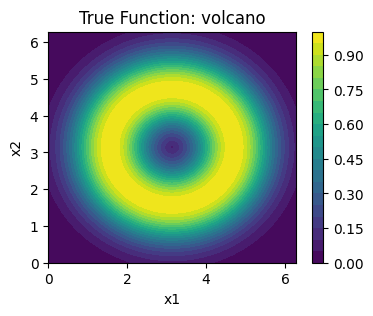}
\subcaption*{Underlying true function}
\end{minipage}
} \\
\begin{minipage}{0.32\textwidth}
\centering
\includegraphics[width=\textwidth]{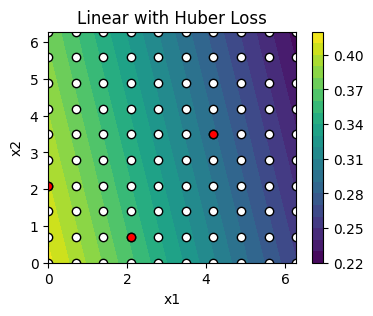}
\end{minipage}
\begin{minipage}{0.32\textwidth}
\centering
\includegraphics[width=\textwidth]{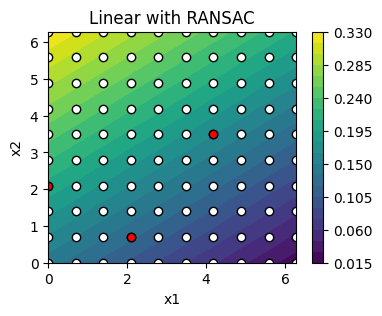}
\end{minipage}
\begin{minipage}{0.32\textwidth}
\centering
\includegraphics[width=\textwidth]{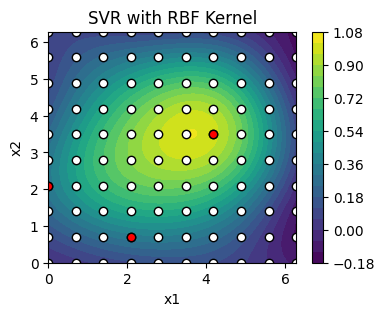}
\end{minipage}
\begin{minipage}{0.32\textwidth}
\centering
\includegraphics[width=\textwidth]{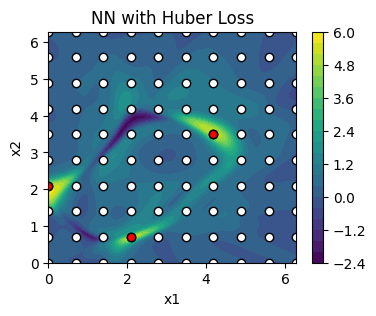}
\end{minipage}
\begin{minipage}{0.32\textwidth}
\centering
\includegraphics[width=\textwidth]{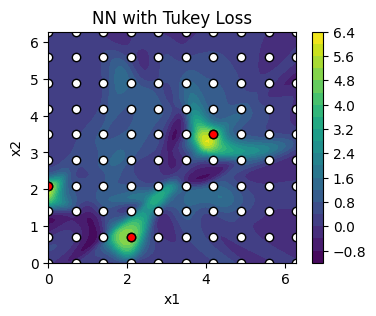}
\end{minipage}
\begin{minipage}{0.32\textwidth}
\centering
\includegraphics[width=\textwidth]{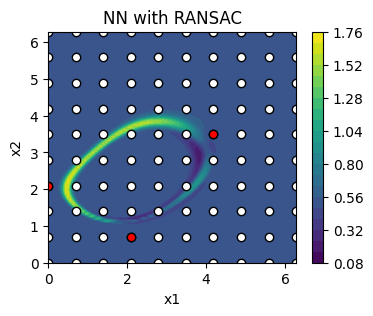}
\end{minipage}
\begin{minipage}{0.32\textwidth}
\centering
\includegraphics[width=\textwidth]{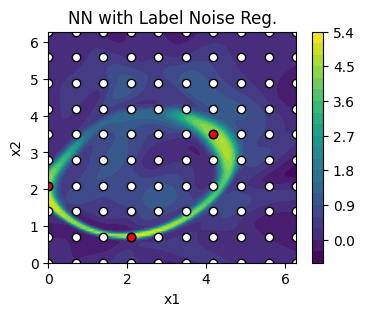}
\end{minipage}
\begin{minipage}{0.32\textwidth}
\centering
\includegraphics[width=\textwidth]{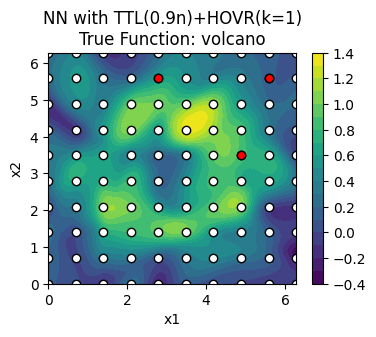}
\end{minipage}
\begin{minipage}{0.32\textwidth}
\centering
\includegraphics[width=\textwidth]{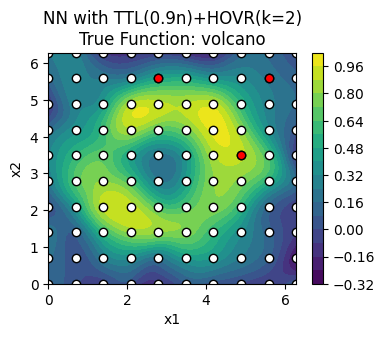}
\end{minipage}
\caption{Volcano function $f_2(x)=\exp\left( -\left( \{x_1 - \pi\}^2 + \{x_2 - \pi\}^2 - 1 \right)^2 \right)$.}
\label{fig:supp_volcano}
\end{figure}

\begin{figure}[!p]
\centering
\fbox{
\begin{minipage}{0.32\textwidth}
\centering
\includegraphics[width=\textwidth]{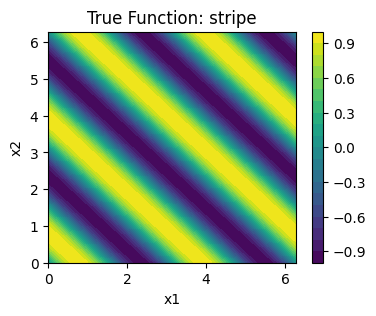}
\subcaption*{Underlying true function}
\end{minipage}
} \\
\begin{minipage}{0.32\textwidth}
\centering
\includegraphics[width=\textwidth]{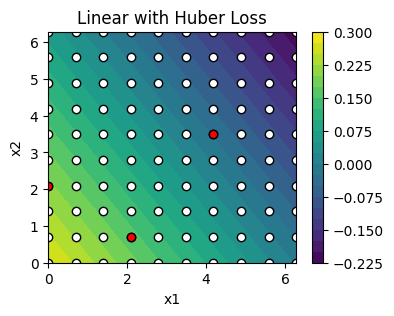}
\end{minipage}
\begin{minipage}{0.32\textwidth}
\centering
\includegraphics[width=\textwidth]{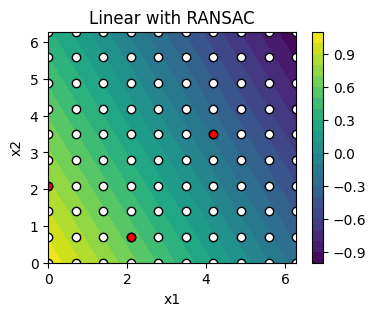}
\end{minipage}
\begin{minipage}{0.32\textwidth}
\centering
\includegraphics[width=\textwidth]{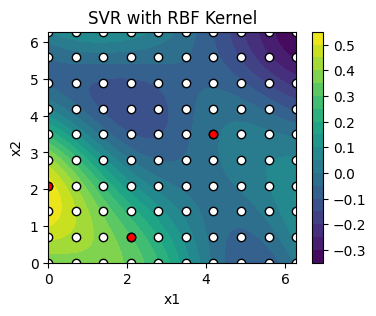}
\end{minipage}
\begin{minipage}{0.32\textwidth}
\centering
\includegraphics[width=\textwidth]{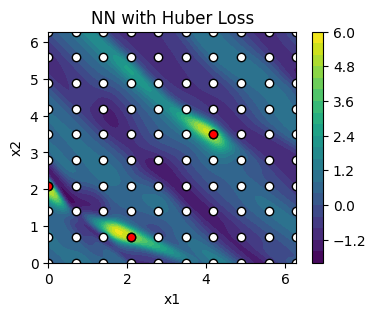}
\end{minipage}
\begin{minipage}{0.32\textwidth}
\centering
\includegraphics[width=\textwidth]{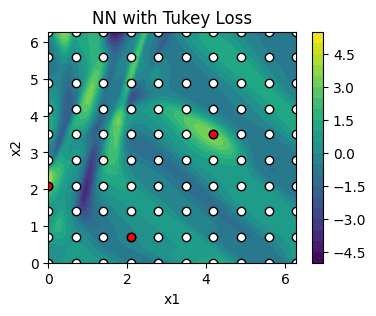}
\end{minipage}
\begin{minipage}{0.32\textwidth}
\centering
\includegraphics[width=\textwidth]{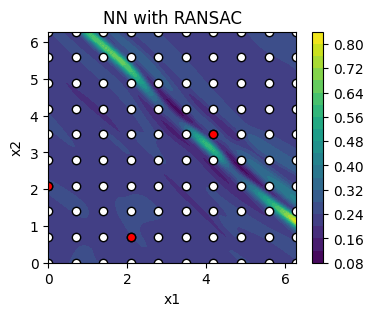}
\end{minipage}
\begin{minipage}{0.32\textwidth}
\centering
\includegraphics[width=\textwidth]{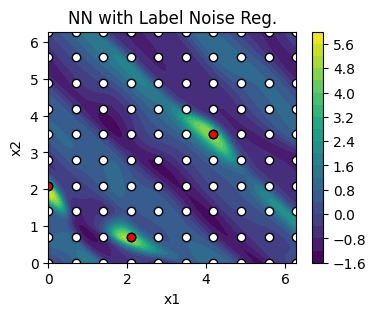}
\end{minipage}
\begin{minipage}{0.32\textwidth}
\centering
\includegraphics[width=\textwidth]{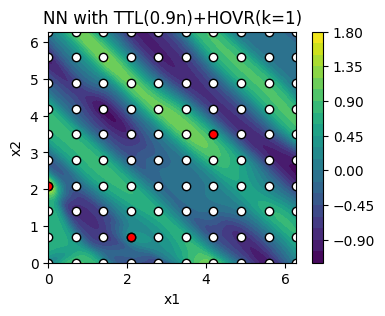}
\end{minipage}
\begin{minipage}{0.32\textwidth}
\centering
\includegraphics[width=\textwidth]{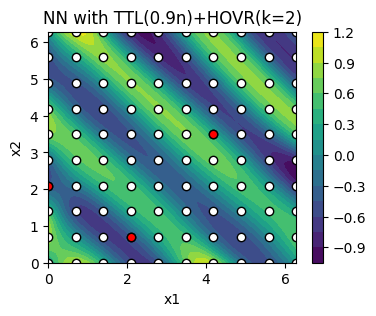}
\end{minipage}
\caption{Stripe function $f_3(x)=\sin(2(x_1+x_2))$.}
\label{fig:supp_stripe}
\end{figure}

\begin{figure}[!p]
\centering
\fbox{
\begin{minipage}{0.32\textwidth}
\centering
\includegraphics[width=\textwidth]{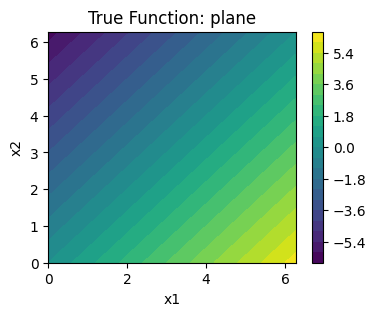}
\subcaption*{Underlying true function}
\end{minipage}
} \\
\begin{minipage}{0.32\textwidth}
\centering
\includegraphics[width=\textwidth]{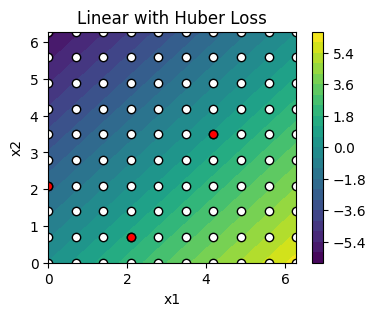}
\end{minipage}
\begin{minipage}{0.32\textwidth}
\centering
\includegraphics[width=\textwidth]{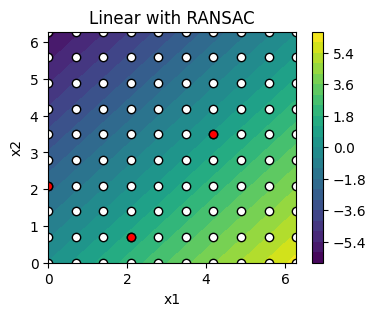}
\end{minipage}
\begin{minipage}{0.32\textwidth}
\centering
\includegraphics[width=\textwidth]{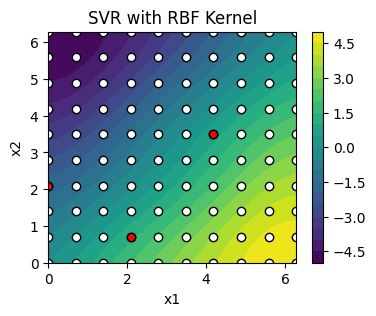}
\end{minipage}
\begin{minipage}{0.32\textwidth}
\centering
\includegraphics[width=\textwidth]{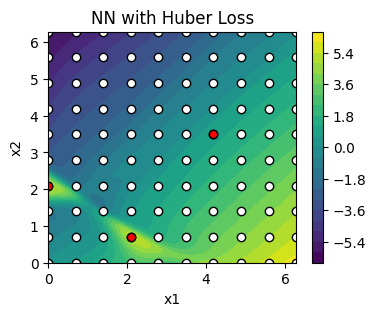}
\end{minipage}
\begin{minipage}{0.32\textwidth}
\centering
\includegraphics[width=\textwidth]{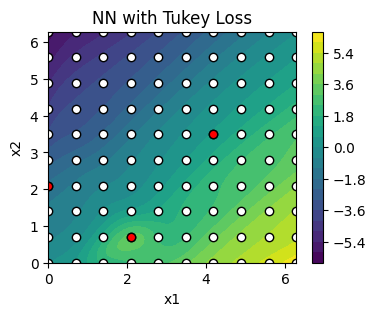}
\end{minipage}
\begin{minipage}{0.32\textwidth}
\centering
\includegraphics[width=\textwidth]{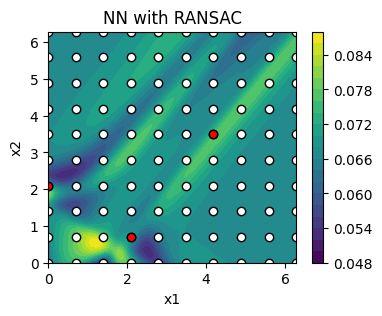}
\end{minipage}
\begin{minipage}{0.32\textwidth}
\centering
\includegraphics[width=\textwidth]{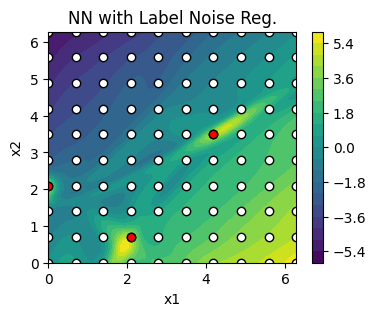}
\end{minipage}
\begin{minipage}{0.32\textwidth}
\centering
\includegraphics[width=\textwidth]{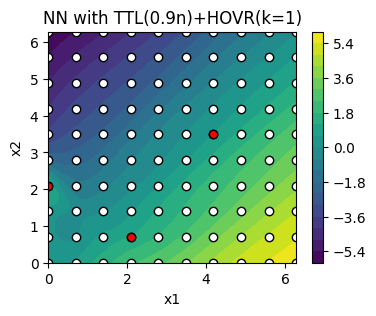}
\end{minipage}
\begin{minipage}{0.32\textwidth}
\centering
\includegraphics[width=\textwidth]{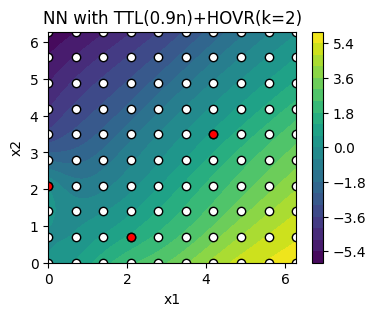}
\end{minipage}
\caption{Plane function $f_4(x)=x_1-x_2$.}
\label{fig:supp_plane}
\end{figure}

\end{document}